\documentclass{svjour3}
\smartqed
\usepackage{newtxtext,newtxmath}
\usepackage{graphicx}
  \graphicspath{{./Pics/}}
  \DeclareGraphicsExtensions{.pdf}
\usepackage{graphbox}
\usepackage{hyperref}
\usepackage{booktabs}
  \newcommand{\ra}[1]{\renewcommand{\arraystretch}{#1}}
\usepackage{xcolor}
\usepackage{url}
\usepackage{amsmath}
\usepackage{pifont}
\usepackage[ruled,vlined]{algorithm2e}
  \SetKwFor{RepTimes}{repeat}{times}{end}
\usepackage{xspace}
\usepackage{tikz}
  \usetikzlibrary{arrows,shapes,fit}
  \usetikzlibrary{automata}
  \usetikzlibrary{backgrounds}
  \usetikzlibrary{decorations.pathmorphing}
  \usetikzlibrary{calc}
  \usetikzlibrary{decorations.pathreplacing}
  \usetikzlibrary{shapes}
\usepackage{wrapfig}
\usepackage[textwidth=1.9cm,color=green!10,textsize=footnotesize]{todonotes}

\newcommand{\eps}{\varepsilon}
\DeclareMathOperator{\A}{\mathcal A}
\DeclareMathOperator{\G}{\mathcal G}

\DeclareMathOperator{\D}{\mathcal D}
\DeclareMathOperator{\C}{\mathcal C}
\newcommand{\complclass}[1]{{\sc #1}\xspace}
\newcommand{\Log}{\complclass{L}}
\newcommand{\NL}{\complclass{NL}}
\newcommand{\coNP}{\complclass{coNP}}
\newcommand{\NP}{\complclass{NP}}
\newcommand{\PSpace}{\complclass{PSpace}}
\newcommand{\PTime}{\complclass{P}}

\begin{document}
\title{Comparing the Notions of Opacity for Discrete-Event Systems\thanks{Partially supported by the Ministry of Education, Youth and Sports under the INTER-EXCELLENCE project LTAUSA19098 and by the University projects IGA PrF 2020 019 and IGA PrF 2021 022.}}
\author{Ji\v{r}\'{i} Balun \and
  Tom{\' a}{\v s}~Masopust
}

\institute{Ji\v{r}\'{i} Balun
  \at Faculty of Science, Palacky University in Olomouc, Czechia.
  \email{jiri.balun01@upol.cz}
  \and
  Tom{\' a}{\v s}~Masopust \at Faculty of Science, Palacky University in Olomouc, Czechia.
  \email{tomas.masopust@upol.cz}
}

\date{Received: date / Accepted: date}
\maketitle

\begin{abstract}
  Opacity is an information flow property characterizing whether a system reveals its secret to a passive observer. Several notions of opacity have been introduced in the literature. We study the notions of language-based opacity, current-state opacity, initial-state opacity, initial-and-final-state opacity, K-step opacity, and infinite-step opacity. Comparing the notions is a natural question that has been investigated and summarized by Wu and Lafortune, who provided transformations among current-state opacity, initial-and-final-state opacity, and language-based opacity, and, for prefix-closed languages, also between language-based opacity and initial-state opacity.
  We extend these results by showing that all the discussed notions of opacity are transformable to each other. Besides a deeper insight into the differences among the notions, the transformations have applications in complexity results. In particular, the transformations are computable in polynomial time and preserve the number of observable events and determinism, and hence the computational complexities of the verification of the notions coincide. We provide a complete and improved complexity picture of the verification of the discussed notions of opacity, and improve the algorithmic complexity of deciding language-based opacity, infinite-step opacity, and K-step opacity.
\end{abstract}

\section{Introduction}
  Applications often require to keep some information about the behavior of a system secret. Properties that guarantee such requirements include anonymity~\cite{SchneiderS96}, noninterference~\cite{Hadj-Alouane05}, secrecy~\cite{Alur2006}, security~\cite{Focardi94}, and opacity~\cite{Mazare04}.

  In this paper, we are interested in opacity for discrete-event systems (DESs) modeled by finite automata. Opacity is a state-estimation property that asks whether a system prevents an intruder from revealing the secret. The intruder is modeled as a passive observer with the complete knowledge of the structure of the system, but with only limited observation of the behavior of the system. Based on the observation, the intruder estimates the behavior of the system, and the system is opaque if the intruder never reveals the secret. In other words, for any secret behavior of the system, there is a non-secret behavior of the system that looks the same to the intruder.

  If the secret is modeled as a set of states, the opacity is referred to as state-based. Bryans et al.~\cite{Bryans2005} introduced state-based opacity for systems modeled by Petri nets, Saboori and Hadjicostis~\cite{SabooriHadjicostis2007} adapted it to (stochastic) automata, and Bryans et al.~\cite{BryansKMR08} generalized it to transition systems.
  If the secret is modeled as a set of behaviors, the opacity is referred to as language-based. Language-based opacity was introduced by Badouel et al.~\cite{Badouel2007} and Dubreil et al.~\cite{Bubreil2008}.
  For more details, we refer the reader to the overview by Jacob et al.~\cite{JacobLF16}.

  Several notions of opacity have been introduced in the literature. In this paper, we are interested in the notions of current-state opacity (CSO), initial-state opacity (ISO), initial-and-final-state opacity (IFO), language-based opacity (LBO), K-step opacity (K-SO), and infinite-step opacity (INSO).
  Current-state opacity is the property that the intruder can never decide whether the system is currently in a secret state.
  Initial-state opacity is the property that the intruder can never reveal whether the computation started in a secret state.
  Initial-and-final-state opacity of Wu and Lafortune~\cite{WuLafortune2013} is a generalization of both, where the secret is represented as a pair of an initial and a marked state. Consequently, initial-state opacity is a special case of initial-and-final-state opacity where the marked states do not play a role, and current-state opacity is a special case where the initial states do not play a role.

  While initial-state opacity prevents the intruder from revealing, at any time during the computation, whether the system started in a secret state, current-state opacity prevents the intruder only from revealing whether the current state of the system is a secret state.
  However, it may happen that the intruder realizes in the future that the system was in a secret state at some former point of the computation. For instance, if the intruder estimates that the system is in one of two possible states and, in the next step, the system proceeds by an observable event that is possible only from one of the states, then the intruder reveals the state in which the system was one step ago.

  This issue has been considered in the literature and led to the notions of K-step opacity (K-SO) and infinite-step opacity (INSO) introduced by Saboori and Hadjicostis~\cite{SabooriHadjicostis2007,SabooriH12a}. While K-step opacity requires that the intruder cannot ascertain the secret in the current and $K$ subsequent states, infinite-step opacity requires that the intruder can never ascertain that the system was in a secret state. Notice that 0-step opacity coincides with current-state opacity by definition, and that an $n$-state automaton is infinite-step opaque if and only if it is $(2^n-2)$-step opaque~\cite{Yin2017}.

  Comparing different notions of opacity for automata models, Saboori and Hadjicostis~\cite{Saboori2008} provided a language-based definition of initial-state opacity, Cassez et al.~\cite{Cassez2012} transformed language-based opacity to current-state opacity, and Wu and Lafortune showed that current-state opacity, initial-and-final-state opacity, and language-based opacity can be transformed to each other. They further provided transformations of initial-state opacity to language-based opacity and to initial-and-final-state opacity, and, for prefix-closed languages, a transformation of language-based opacity to initial-state opacity.

  In this paper, we extend these results by showing that, for automata models, all the discussed notions of opacity are transformable to each other. As well as the existing transformations, our transformations are computable in polynomial time and preserve the number of observable events and determinism (whenever it is meaningful). In more detail, the transformations of Wu and Lafortune~\cite{WuLafortune2013} preserve the determinism of transitions, but result in automata with a set of initial states. This issue can, however, be easily fixed by adding a new initial state, connecting it to the original initial states by new unobservable events, and making the original initial states non-initial.
  We summarize our results, together with the existing results, in Fig.~\ref{fig:opacity-overview}.

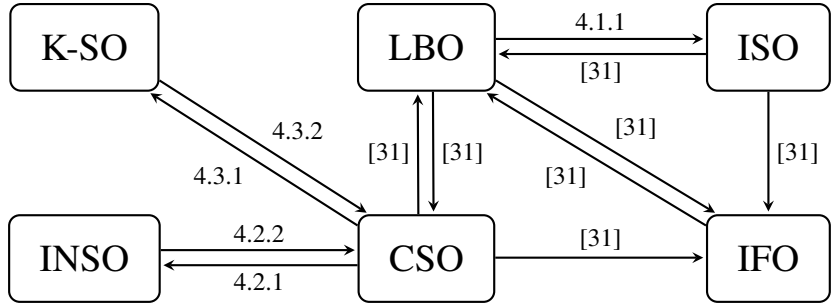
\begin{figure}
  \centering
  \begin{tikzpicture}[baseline,->,>=stealth,auto,shorten >=1pt,node distance=4.5cm,font=\Large,
    state/.style={rectangle,rounded corners,minimum size=6mm,inner sep=4mm,draw=black,thick,initial text=}]
    \node[state] (lbo)               {LBO};
    \node[state] (kso) [left of=lbo,node distance=4.5cm]   {K-SO};
    \node[state] (iso) [right of=lbo]  {ISO};
    \node[state] (cso) [below of=lbo,node distance=2.8cm]  {CSO};
    \node[state] (inso) [left of=cso,node distance=4.5cm]   {INSO};
    \node[state] (ifo) [right of=cso]  {IFO};

    \path[thick]
    ([xshift=10mm] lbo.south west) edge node[font=\normalsize] {\cite{WuLafortune2013}} ([xshift=10mm] cso.north west)
    ([xshift=8mm] cso.north west) edge node[font=\normalsize] {\cite{WuLafortune2013}} ([xshift=8mm] lbo.south west)
    ([yshift=7mm] lbo.south east) edge node[font=\normalsize] {\ref{redLBOtoISO}} ([yshift=7mm] iso.south west)
    ([yshift=5mm] iso.south west) edge node[font=\normalsize] {\cite{WuLafortune2013}} ([yshift=5mm] lbo.south east)

    ([yshift=7mm] inso.south east) edge node[font=\normalsize] {\ref{redINSOtoCSO}} ([yshift=7mm] cso.south west)
    ([yshift=5mm] cso.south west) edge node[font=\normalsize] {\ref{redCSOtoINSO}} ([yshift=5mm] inso.south east)
    ([yshift=6mm] cso.south east) edge node[font=\normalsize] {\cite{WuLafortune2013}} ([yshift=6mm] ifo.south west)

    ([yshift=1.5mm] lbo.south east) edge node[font=\normalsize] {\cite{WuLafortune2013}} ([xshift=1.5mm] ifo.north west)
    ([yshift=-1.5mm] ifo.north west) edge node[font=\normalsize] {\cite{WuLafortune2013}} ([xshift=-1.5mm] lbo.south east)

    ([yshift=1.5mm] kso.south east) edge node[font=\normalsize] {\ref{redKSOtoCSO}} ([xshift=1.5mm] cso.north west)
    ([yshift=-1.5mm] cso.north west) edge node[font=\normalsize] {\ref{redCSOtoKSO}} ([xshift=-1.5mm] kso.south east)
    (iso) edge node[font=\normalsize] {\cite{WuLafortune2013}} (ifo);
  \end{tikzpicture}
  \caption{Overview of the transformations among the notions of opacity for automata models.}
  \label{fig:opacity-overview}
\end{figure}

  There are two immediate applications of the transformations. First, the transformations provide a deeper understanding of the differences among the opacity notions from the structural point of view. For instance, the reader may deduce from the transformations that, for prefix-closed languages, the notions of language-based opacity, initial-state opacity, and current-state opacity coincide, or that to transform current-state opacity to infinite-step opacity means to add only a single state and a few transitions.

  Second, the transformations provide a tool to obtain the complexity results for all the discussed opacity notions by studying just one of the notions. For an illustration, consider, for instance, our recent result showing that deciding current-state opacity for systems modeled by DFAs with three events, one of which is unobservable, is \PSpace-complete~\cite{BalunMasopustIFACWC2020}. Since we can transform the problems of deciding current-state opacity and of deciding infinite-step opacity to each other in polynomial time, preserving determinism and the number of observable events, we obtain that deciding infinite-step opacity for systems modeled by DFAs with three events, one of which is unobservable, is \PSpace-complete as well. In particular, combining the transformations with known results~\cite{BalunMasopustIFACWC2020,JacobLF16}, we obtain a complete complexity picture of the verification of the discussed notions of opacity as summarized in Table~\ref{table01}.

  The fact that checking opacity for DESs is \PSpace-complete was known for some of the considered notions~\cite{JacobLF16}. In particular, deciding current-state opacity, initial-state opacity, and language-based opacity were known to be \PSpace-complete, deciding K-step opacity was known to be \NP-hard, and deciding infinite-step opacity was known to be \PSpace-hard.
  Complexity theory tells us that any two \PSpace-complete problems can be transformed to each other in polynomial time. In other words, it gives the existence of polynomial transformations between the notions of opacity for which the verification is \PSpace-complete. However, the theory and the \PSpace-hardness proofs presented in the literature do not give a clue how to obtain these transformations. Therefore, from the complexity point of view, our contribution is not the existence of the transformations, but the construction of specific transformations. Since the presented transformations preserve determinism and the number of observable events, they allow us to present stronger results than those known in the literature~\cite{JacobLF16} that we summarize in Table~\ref{table01}.

  The transformations further allow us to improve the algorithmic complexity of deciding language-based opacity, infinite-step opacity, and K-step opacity, although we do not use the transformations themselves, but rather the deeper insight into the problems they provide. For language-based opacity, Lin~\cite{Lin2011} suggested an algorithm with complexity $O(2^{2n})$, where $n$ is the number of states of the input automaton. In this paper, we improve this complexity to $O((n + m\ell) 2^n)$, where $\ell=|\Sigma_o|$ is the number of observable events and $m \le \ell n^2$ is the number of transitions in the projected automaton of the input automaton. For infinite-step opacity and K-step opacity, the latest results are by Yin and Lafortune~\cite{Yin2017} who designed an algorithm for checking infinite-step opacity with complexity $O(\ell 2^{2n})$, and an algorithm for checking K-step opacity with complexity $O(\min\{\ell 2^{2n},\ell^{K+1} 2^{n}\})$. In this paper, we suggest a new algorithm for deciding infinite-step opacity with complexity $O((n + m\ell) 2^n)$, and a new algorithm for checking K-step opacity with complexity $O((K+1)2^n (n + m\ell^2))$. Notice that $K$ is bounded by $2^n-2$, since an $n$-state automaton is infinite-step opaque if and only if it is $(2^n-2)$-step opaque~\cite{Yin2017}. Consequently, our algorithm improves the complexity if $K$ is either very large (larger than $2^n-2$) or polynomial with respect to $n$; otherwise, the two-way observer technique of Yin and Lafortune~\cite{Yin2017} is more efficient, and it is a challenging open problem whether its complexity can be further improved.
  All our results are summarized in Table~\ref{table01}.

  \begin{table}\centering
    \ra{1.2}
    \begin{tabular}{@{}llll@{}}\toprule
              Opacity notion
            & $|\Sigma_o|=1$
            & $|\Sigma_o|\ge 2$
            & Order\\
          \midrule
          CSO       & \coNP-complete~\cite{BalunMasopustIFACWC2020}
                    & \PSpace-complete~\cite{BalunMasopustIFACWC2020}
                    & $O(\ell 2^{n})$~\cite{Saboori2011}\\
          LBO       & \coNP-complete
                    & \PSpace-complete
                    & $O((n+m\ell) 2^n)$ (Thm~\ref{lboAlgo})\\
          ISO       & \NL-complete (Thm~\ref{NLunaryISO})
                    & \PSpace-complete
                    & $O(\ell 2^{n})$~\cite{WuLafortune2013}\\
          IFO       & \coNP-complete
                    & \PSpace-complete
                    & $O(\ell 2^{2n})$~\cite{WuLafortune2013}\\
          K-SO      & \coNP-complete
                    & \PSpace-complete
                    & $O((K+1)2^n (n + m \ell^2))$ (Sec~\ref{ksoAlgo})
                    \\
          INSO      & \coNP-complete
                    & \PSpace-complete
                    & $O((n+m\ell) 2^n)$ (Sec~\ref{insoAlgo}) \\
      \bottomrule
    \end{tabular}
      \caption{Complexity of verifying the notions of opacity for DESs with $\Sigma_o$ being the set of observable events following from the transformations and known results; $n$ stands for the number of states of the input automaton, $\ell$ for the number of observable events of the input automaton, and $m \le \ell n^2$ for the number of transitions in the projected automaton of the input automaton.}
      \label{table01}
  \end{table}

\section{Preliminaries}
  We assume that the reader is familiar with the basic notions of automata theory~\cite{Lbook}.
  For a set $S$, $|S|$ denotes the cardinality of $S$, and $2^{S}$ the power set of $S$. Let $\mathbb{N}$ denote the set of all non-negative integers. An alphabet $\Sigma$ is a finite nonempty set of events. A string over $\Sigma$ is a sequence of events from $\Sigma$. Let $\Sigma^*$ denote the set of all finite strings over $\Sigma$; the empty string is denoted by $\varepsilon$. A language $L$ over $\Sigma$ is a subset of $\Sigma^*$. The set of all prefixes of strings of $L$ is the set $\overline{L}=\{u \mid \text{there is } v\in \Sigma^* \text{ such that } uv \in L\}$. For a string $u \in \Sigma^*$, $|u|$ denotes the length of $u$, and $\overline{u}$ denotes the set of all prefixes of $u$.

  A {\em nondeterministic finite automaton\/} (NFA) over an alphabet $\Sigma$ is a structure $\G = (Q,\Sigma,\delta,I,F)$, where $Q$ is a finite  set of states, $I\subseteq Q$ is a set of initial states, $F \subseteq Q$ is a set of marked states, and $\delta \colon Q\times\Sigma \to 2^Q$ is a transition function that can be extended to the domain $2^Q\times\Sigma^*$ by induction. To simplify our proofs, we use the notation $\delta(Q,S) = \cup_{s\in S}\,\delta(Q, s)$, where $S\subseteq \Sigma^*$.
  For a set of states $Q_0\subseteq Q$, the language marked by $\G$ from the states of $Q_0$ is the set $L_m(\G,Q_0) = \{w\in \Sigma^* \mid \delta(Q_0,w)\cap F \neq\emptyset\}$, and the language generated by $\G$ from the states of $Q_0$ is the set $L(\G,Q_0) = \{w\in \Sigma^* \mid \delta(Q_0,w)\neq\emptyset\}$. The language marked by $\G$ is then $L_m(\G)=L_m(\G,I)$, and the language generated by $\G$ is $L(\G)=L(\G,I)$.
  The NFA $\G$ is {\em deterministic\/} (DFA) if $|I|=1$ and $|\delta(q,a)|\le 1$ for every $q\in Q$ and $a \in \Sigma$.
  An automaton $\G$ is {\em non-blocking\/} if $\overline{L_m(\G)} = L(\G)$.

  A {\em discrete-event system\/} (DES) $G$ over $\Sigma$ is an NFA together with the partition of the alphabet $\Sigma$ into two disjoint subsets $\Sigma_o$ and $\Sigma_{uo}=\Sigma\setminus\Sigma_o$ of {\em observable\/} and {\em unobservable events}, respectively. In the case where all states of the automaton are marked, we simply write $G=(Q,\Sigma,\delta,I)$ without specifying the set of marked states.

  When discussing the state estimation properties, the literature often studies deterministic systems with a set of initial states. Such systems are known as deterministic DES and defined as a DFA with several initial states; namely, a {\em deterministic\/} DES is an NFA $\G=(Q,\Sigma,\delta,I,F)$, where $|\delta(q,a)|\le 1$ for every $q\in Q$ and $a \in \Sigma$.

  The opacity property is based on partial observations of events described by {\em projection\/} $P\colon \Sigma^* \to \Sigma_o^*$. The projection is a morphism defined by $P(a) = \varepsilon$ for $a\in \Sigma_{uo}$, and $P(a)= a$ for $a\in \Sigma_o$. The action of $P$ on a string $\sigma_1\sigma_2\cdots\sigma_n$, with $\sigma_i \in \Sigma$ for $1\le i\le n$, is to erase all events that do not belong to $\Sigma_o$, that is, $P(\sigma_1\sigma_2\cdots\sigma_n)=P(\sigma_1) P(\sigma_2) \cdots P(\sigma_n)$. The definition can be readily extended to languages.


  Let $G$ be a NFA over $\Sigma$, and let $P\colon \Sigma^* \to \Sigma_o^*$ be a projection. By the {\em projected automaton\/} of $G$, we mean the automaton $P(G)$ obtained from $G$ by replacing every transition $(p,a,q)$ by the transition $(p,P(a),q)$, and by eliminating the $\eps$-transitions. In particular, if $\delta$ is the transition function of $G$, then the transition function $\gamma$ of the automaton $P(G)$ is defined as $\gamma(q,a)=\hat{\delta}(q,a)$, where $\hat{\delta}\colon Q\times \Sigma^* \to 2^Q$ is the extension of $\delta$ to the domain $Q\times \Sigma^*$, that is, $\hat{\delta}(q,\eps)=\{q\}$ and $\hat{\delta}(q,wa)=\bigcup_{p\in\hat{\delta}(q,w)} \delta(p,a)$. Then $P(G)$ is an NFA over $\Sigma_o$, with the same set of states as $G$, that recognizes the language $P(L_m(G))$ and can be constructed in polynomial time~\cite{HopcroftU79}.
  The DFA constructed from $P(G)$ by the subset construction is called an {\em observer}~\cite{Lbook}. In the worst case, the observer has exponentially many states compared with the automaton $G$~\cite{JiraskovaM12,wong98}.

  A {\em decision problem\/} is a yes-no question. A decision problem is {\em decidable\/} if there is an algorithm that solves it. Complexity theory classifies decidable problems into classes based on the time or space an algorithm needs to solve the problem. The complexity classes we consider are \Log, \NL, \PTime, \NP, and \PSpace denoting the classes of problems solvable by a deterministic logarithmic-space, nondeterministic logarithmic-space, deterministic polynomial-time, nondeterministic polynomial-time, and deterministic polynomial-space algorithm, respectively. The hierarchy of classes is \Log $\subseteq$ \NL $\subseteq$ \PTime $\subseteq$ \NP $\subseteq$ \PSpace. Which of the inclusions are strict is an open problem. The widely accepted conjecture is that all are strict. A decision problem is \NL-complete (resp. \NP-complete, \PSpace-complete) if (i) it belongs to \NL (resp. \NP, \PSpace) and (ii) every problem from \NL (resp. \NP, \PSpace) can be reduced to it by a deterministic logarithmic-space (resp. polynomial-time) algorithm. Condition (i) is called {\em membership\/} and condition (ii) {\em hardness}.

\section{Notions of Opacity}
  In this section, we recall the definitions of the notions of opacity we discuss. The notion of initial-and-final-state opacity is recalled to make the paper self-contained.

  Current-state opacity asks whether the intruder cannot decide, at any instance of time, whether the system is currently in a secret state.

  \begin{definition}[Current-state opacity (CSO)]
    Given a DES $G=(Q,\Sigma,\delta,I)$, a projection $P\colon\Sigma^*\to \Sigma_o^*$, a set of secret states $Q_S \subseteq Q$, and a set of non-secret states $Q_{NS} \subseteq Q$. System $G$ is {\em current-state opaque\/} if for every string $w$ such that $\delta(I,w)\cap Q_S \neq \emptyset$, there exists a string $w'$ such that $P(w)=P(w')$ and $\delta(I,w')\cap Q_{NS} \neq \emptyset$.
  \end{definition}

  The definition of current-state opacity can be reformulated as a language inclusion as shown in the following lemma. This result is similar to that of Wu and Lafortune \cite{WuLafortune2013} used to transform current-state opacity to language-based opacity. We use this alternative definition to simplify proofs.
  \begin{lemma}[\cite{BalunMasopustIFACWC2020}]\label{CSOinclusion}
    Let $G=(Q,\Sigma,\delta,I)$ be a DES, $P\colon\Sigma^*\to \Sigma_o^*$ a projection, and $Q_S, Q_{NS} \subseteq Q$ sets of secret and non-secret states, respectively. 
    Let $G_{S}=(Q,\Sigma,\delta,I,Q_S)$ and $G_{NS}=(Q,\Sigma,\delta,I,Q_{NS})$, then $G$ is current-state opaque if and only if $L_m(P(G_{S})) \subseteq L_m(P(G_{NS}))$.
  \end{lemma}

  The second notion of opacity under consideration is language-based opacity. Intuitively, a system is language-based opaque if for any string $w$ in the secret language, there exists a string $w'$ in the non-secret language with the same observation $P(w)=P(w')$. In this case, the intruder cannot conclude whether the secret string $w$ or the non-secret string $w'$ has occurred. We recall the most general definition by Lin~\cite{Lin2011}.
  \begin{definition}[Language-based opacity (LBO)]
    Given a DES $G=(Q,\Sigma,\delta,I)$, a projection $P\colon\Sigma^*\to \Sigma_o^*$, a secret language $L_S \subseteq L(G)$, and a non-secret language $L_{NS} \subseteq L(G)$. System $G$ is {\em language-based opaque\/} if $L_S \subseteq P^{-1}P(L_{NS})$.
  \end{definition}

  It is worth mentioning that the secret and non-secret languages are often considered to be regular; and we consider it as well. The reason is that, for non-regular languages, the inclusion problem is undecidable; see Asveld and Nijholt~\cite{AsveldN00} for more details.

  The third notion is the notion of initial-state opacity. Initial-state opacity asks whether the intruder can never reveal whether the computation started in a secret state.
  \begin{definition}[Initial-state opacity (ISO)]
    Given a DES $G=(Q,\Sigma,\delta,I)$, a projection $P\colon \Sigma^*\rightarrow \Sigma_o^*$, a set of secret initial states $Q_S\subseteq I$, and a set of non-secret initial states $Q_{NS}\subseteq I$. System $G$ is \textit{initial-state opaque\/} with respect to $Q_S$, $Q_{NS}$ and $P$ if for every $w \in L(G,Q_S)$, there exists $w' \in L(G,Q_{NS})$ such that $P(w) = P(w')$.
  \end{definition}

  The fourth notion is the notion of initial-and-final-state opacity of Wu and Lafortune~\cite{WuLafortune2013}. Initial-and-final-state opacity is a generalization of both current-state opacity and initial-state opacity, where the secret is represented as a pair of an initial and a marked state. Consequently, initial-state opacity is a special case of initial-and-final-state opacity where the marked states do not play a role, and current-state opacity is a special case where the initial states do not play a role.

  \begin{definition}[Initial-and-final-state opacity (IFO)]
    Given a DES $G=(Q,\Sigma,\delta,I)$, a projection $P\colon \Sigma^*\rightarrow \Sigma_o^*$, a set of secret state pairs $Q_S\subseteq I \times Q$, and a set of non-secret state pairs $Q_{NS}\subseteq I \times Q$. System $G$ is \textit{initial-and-final-state opaque\/} with respect to $Q_S$, $Q_{NS}$ and $P$ if for every secret pair $(q_0,q_f) \in Q_S$ and every $w \in L(G,q_0)$ such that $q_f \in \delta(q_0, w)$, there exists $(q_0',q_f') \in Q_{NS}$ and $w' \in L(G,q_0')$ such that $q_f' \in \delta(q_0', w')$ and $P(w) = P(w')$.
  \end{definition}

  The fifth notion is the notion of K-step opacity. K-step opacity is a generalization of current-state opacity requiring that the intruder cannot reveal the secret in the current and $K$ subsequent states. By definition, current-state opacity is equivalent to 0-step opacity.
  We slightly generalize and reformulate the definition of Saboori and Hadjicostis~\cite{SabooriH12a}.
  \begin{definition}[K-step opacity (K-SO)]
    Given a system $G=(Q,\Sigma,\delta,I)$, a projection $P \colon \Sigma^* \to \Sigma_o^*$, a set of secret states $Q_S\subseteq Q$, a set of non-secret states $Q_{NS}\subseteq Q$, and a non-negative integer $K\in\mathbb{N}$. System $G$ is \textit{K-step opaque} with respect to $Q_S$, $Q_{NS}$, and $P$ if for every string $st \in L(G)$ such that $|P(t)| \leq K$ and $\delta(\delta(I, s)\cap Q_S, t) \neq \emptyset$, there exists a string $s't' \in L(G)$ such that $P(s)=P(s')$, $P(t)=P(t')$, and $\delta (\delta(I,s')\cap Q_{NS}, t') \neq \emptyset$.
  \end{definition}

  Finally, the last notion we consider is the notion of infinite-step opacity. Infinite-step opacity is a further generalization of K-step opacity by setting $K$ being infinity. Actually, Yin and Lafortune~\cite{Yin2017} have shown that an $n$-state automaton is infinite-step opaque if and only if it is $(2^n-2)$-step opaque.
  Again, we slightly generalize and reformulate the definition of Saboori and Hadjicostis~\cite{SabooriH11}.
  \begin{definition}[Infinite-step opacity (INSO)]
    Given a system $G=(Q,\Sigma,\delta,I)$, a projection $P\colon \Sigma^*\rightarrow \Sigma_o^*$, a set of secret states $Q_S\subseteq Q$, and a set of non-secret states $Q_{NS}\subseteq Q$. System $G$ is \textit{infinite-step opaque} with respect to $Q_S$, $Q_{NS}$ and $P$ if for every string $st \in L(G)$ such that $\delta(\delta(I, s)\cap Q_S, t) \neq \emptyset$, there exists a string $s't' \in L(G)$ such that $P(s)=P(s')$, $P(t)=P(t')$, and $\delta (\delta(I,s')\cap Q_{NS}, t') \neq \emptyset$.
  \end{definition}

\section{Transformations}
  Although some of the transformations were previously known in the literature, Wu and Lafortune~\cite{WuLafortune2013} were first who studied the transformations systematically. In particular, they provided polynomial-time transformations among current-state opacity, language-based opacity, initial-state opacity, and initial-and-final-state opacity, see Fig.~\ref{fig:opacity-overview}. Inspecting the reductions, it can be seen that after  eliminating the unnecessary \texttt{Trim} operations, the transformations use only logarithmic space, preserve the number of observable events, and determinism (whenever it is meaningful). As we already pointed out, the transformations of Wu and Lafortune~\cite{WuLafortune2013} preserve the determinism of transitions, but they admit a set of initial states. This issue can, however, be easily eliminated by adding a new initial state, connecting it to the original initial states by new unobservable events, and making the original initial states non-initial.
  However, their transformation from language-based opacity to initial-state opacity is restricted only to the case where the secret and non-secret languages of the language-based opacity problem are prefix closed.

  We complete the polynomial-time transformations among all the discussed notions of opacity. In particular, we provide a general transformation from language-based opacity to initial-state opacity in Section~\ref{redLBOtoISO}, transformations between infinite-step opacity and current-state opacity in Section~\ref{redINSO-CSO}, and transformations between K-step opacity and current-state opacity in Section~\ref{redKSO-CSO}.
  All the transformations preserve the number of observable events and determinism. Except for a few exceptions, the transformations need only logarithmic space.
  Our results are summarized in Fig.~\ref{fig:opacity-overview} with references to the corresponding sections.

  The following auxiliary lemma states that we can reduce the number of observable events in DESs with at least three observable events without affecting current-state opacity and initial-state opacity of the DES. We make use of this lemma to preserve the number of observable events in cases where we introduce new  observable events in our reductions, namely in Sections~\ref{redLBOtoISO}, \ref{redINSOtoCSO}, and~\ref{redKSOtoCSO}.

  \begin{figure}
    \centering
    \includegraphics[scale=1]{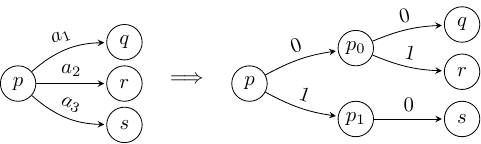}
    \caption{The replacement of three observable events $\{a_1,a_2,a_3\}$ with the encoding $e(a_1)=00$, $e(a_2)=01$, and $e(a_3)=10$, and new states $p_0$ and $p_1$.}
    \label{binEncoding}
  \end{figure}

  \begin{lemma}\label{binEnc}
    Let $G=(Q,\Sigma,\delta,I,F)$ be an NFA, and let $\Gamma_o\subseteq\Sigma_o$ contain at least three events. Let $G'=(Q',(\Sigma-\Gamma_o)\cup\{0,1\},\delta',I,F)$ be an NFA obtained from $G$ as follows. Let $k=\lceil \log_2(|\Gamma_o|) \rceil$, and let $e\colon \Gamma_o \to \{0,1\}^k$ be a binary encoding of the events of $\Gamma_o$. We replace every transition $(p,a,q)$  with $a\in \Gamma_o$ by $k$ transitions
    \[
      (p,b_1,p_{b_1}), (p_{b_1},b_2,p_{b_1b_2}),\ldots, (p_{b_1\cdots b_{k-1}},b_k,q)
    \]
    where $e(a)=b_1b_2\cdots b_k \in \{0,1\}^k$, and $p_{b_1},\ldots,p_{b_1\cdots b_{k-1}}$ are states that are added to the state set of $G'$. Notice that these states are neither secret nor non-secret and that, to preserve determinism, they are newly created when they are needed for the first time, and reused when they are needed later during the replacements, cf.~Fig.~\ref{binEncoding} illustrating a replacement of three observable events $\{a_1,a_2,a_3\}$ with the encoding $e(a_1)=00$, $e(a_2)=01$, and $e(a_3)=10$.
    Then $G$ is current-state (initial-state) opaque with respect to $Q_S$, $Q_{NS}$, and $P\colon \Sigma^*\to\Sigma_o^*$ if and only if $G'$ is current-state (initial-state) opaque with respect to $Q_{S}$, $Q_{NS}$, and $P'\colon [(\Sigma-\Gamma_o)\cup\{0,1\}]^* \to [(\Sigma_o-\Gamma_o)\cup \{0,1\}]^*$.
  \end{lemma}
  \begin{proof}
    To show that $G$ is current-state opaque if and only if $G'$ is current-state opaque, we define the languages $L_S=L_m(Q,\Sigma,\delta,I,Q_S)$, $L_{NS}=L_m(Q,\Sigma,\delta,I,Q_{NS})$, $L'_S=L_m(Q',(\Sigma-\Gamma_o)\cup\{0,1\},\delta',I,Q_S)$, and $L'_{NS}=L_m(Q',(\Sigma-\Gamma_o)\cup\{0,1\},\delta',I,Q_{NS})$. Using Lemma~\ref{CSOinclusion}, we now need to show that $P(L_{S})\subseteq P(L_{NS})$ if and only if $P'(L'_{S})\subseteq P'(L'_{NS})$. To this end, we define a morphism $f\colon \Sigma^* \to ((\Sigma-\Gamma_o)\cup\{0,1\})^*$ so that $f(a)=e(a)$ for $a\in \Gamma_o$, and $f(a)=a$ for $a\in\Sigma-\Gamma_o$. By the definition of $e$ and the construction of $G'$, for any string $w$, we have that $w \in L(G)$ if and only if $f(w) \in L(G')$. In particular, $P(w)\in P(L_{S})$ if and only if $P'(f(w)) \in P'(L'_{S})$, and $P(w)\in P(L_{NS})$ if and only if $P'(f(w)) \in P'(L'_{NS})$, which completes this part of the proof.

    To show that $G$ is initial-state opaque if and only if $G'$ is initial-state opaque, we define the languages $L_S=L(Q,\Sigma,\delta,Q_S)$, $L_{NS}=L(Q,\Sigma,\delta,Q_{NS})$, $L'_S=L(Q',(\Sigma-\Gamma_o)\cup\{0,1\},\delta',Q_S)$, and $L'_{NS}=L(Q',(\Sigma-\Gamma_o)\cup\{0,1\},\delta',Q_{NS})$. Since this transforms initial-state opacity to language-based opacity~\cite{WuLafortune2013}, it is sufficient to show that $P(L_{S})\subseteq P(L_{NS})$ if and only if $P'(L'_{S})\subseteq P'(L'_{NS})$. However, this can be shown analogously as above.
  \qed\end{proof}

  Notice that this binary encoding can be done in polynomial time, and that it preserves determinism.

\subsection{Transformations between LBO and ISO}
  In this section, we discuss the transformations between language-based opacity and initial-state opacity. The transformation from initial-state opacity to language-based opacity has been provided by Wu and Lafortune~\cite{WuLafortune2013}, as well as the transformation from language-based opacity to initial-state opacity for the case where both the secret and the non-secret language of the language-based opacity problem are prefix closed. We now extend the transformation from language-based opacity to initial-state opacity to the general case.

\subsubsection{Transforming LBO to ISO}\label{redLBOtoISO}
  The language-based opacity problem consists of a DES $G_{LBO}$ over $\Sigma$, a projection $P\colon\Sigma^*\to \Sigma_o^*$, a secret language $L_S \subseteq L(G)$, and a non-secret language $L_{NS} \subseteq L(G)$. We transform it to a DES $G_{ISO}$ in such a way that $G_{LBO}$ is language-based opaque if and only if $G_{ISO}$ is initial-state opaque.

  Assume that the languages $L_S$ and $L_{NS}$ are represented by the non-blocking automata $A_S=(Q_S,\Sigma_S, \delta_S, I_S, F_S)$ and $A_{NS}=(Q_{NS},\Sigma_{NS}, \delta_{NS}, I_{NS}, F_{NS})$, respectively. Without loss of generality, we may assume that their sets of states are disjoint, that is, $Q_{S}\cap Q_{NS} = \emptyset$.

  Our transformation proceeds in two steps:
  \begin{enumerate}
   \item We construct a DES $G_{ISO}$ with one additional observable event $@$.
   \item We use Lemma~\ref{binEnc} to reduce the number of observable events of $G_{ISO}$ by one.
  \end{enumerate}

  \begin{figure}
    \begin{minipage}{\textwidth}
      \centering
      \includegraphics[align=c,scale=.89]{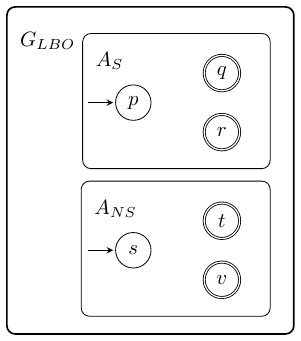}
      \quad $\Longrightarrow$ \quad
      \includegraphics[align=c,scale=.89]{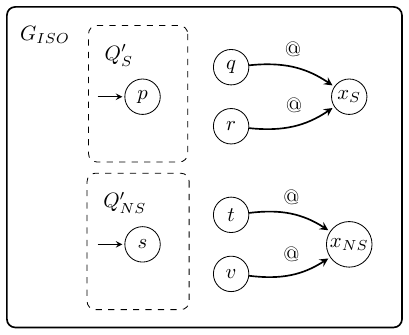}
    \end{minipage}
    \caption{Transforming LBO to ISO.}
    \label{fig:lbo-so}
  \end{figure}

  Since the second step follows from Lemma~\ref{binEnc}, we only describe the first step, that is, the construction of $G_{ISO}$ over $\Sigma\cup\{@\}$, and the specification of the sets of secret states $Q'_{S}$ and non-secret states $Q'_{NS}$.
  From the automata $A_{S}$ and $A_{NS}$, we construct the automata $G_S=(Q_S\cup \{x_S\},\Sigma_S, \delta_S, I_S,Q_S\cup \{x_S\})$ and $G_{NS}=(Q_{NS}\cup \{x_{NS}\},\Sigma_{NS}, \delta_{NS}, I_{NS},Q_{NS}\cup \{x_{NS}\})$ by adding two new states $x_{S}$ and $x_{NS}$, and the following transitions, see Fig.~\ref{fig:lbo-so} for an illustration of the construction:
  \begin{itemize}
    \item for every state $q\in F_S$, we add a new transition $(q,@,x_S)$ to $\delta_S$;
    \item for every state $q\in F_{NS}$, we add a new transition $(q,@,x_{NS})$ to $\delta_{NS}$.
  \end{itemize}

  Let $Q'_S = I_S$ denote the set of secret initial states of $G_{ISO}$, and let $Q'_{NS} = I_{NS}$ denote the set of non-secret initial states of $G_{ISO}$. We extend projection $P$ to $P'\colon (\Sigma\cup\{@\})^* \to (\Sigma_o\cup\{@\})^*$.
  Finally, let $G_{ISO}$ denote the automata $G_S$ and $G_{NS}$ considered as a single NFA.
  Before we show that $G_{LBO}$ is language-based opaque if and only if $G_{ISO}$ is initial-state opaque, notice that the transformation can be done in polynomial time and that it preserves determinism.

  \begin{theorem}\label{thm:lbo-iso}
    The DES $G_{LBO}$ is language-based opaque with respect to $L_{S}$, $L_{NS}$, and $P$ if and only if the DES $G_{ISO}$ is initial-state opaque with respect to $Q'_S$, $Q'_{NS}$, and $P'$.
  \end{theorem}
  \begin{proof}
    We need to show that $P(L_S)\subseteq P(L_{NS})$ if and only if $P'(L(G_S))\subseteq P'(L(G_{NS}))$.
    However, by construction, $L(G_S) =\overline{L_S} \cup L_S@$ and $L(G_{NS})=\overline{L_{NS}} \cup L_{NS}@$, and hence  $P(L_S)\subseteq P(L_{NS})$ if and only if $P'(L(G_S))\subseteq P'(L(G_{NS}))$, which is if and only if $G_{ISO}$ is initial-state opaque.
\qed\end{proof}

  We now provide an illustrative example.
  \begin{example}\label{ex:lbo-iso}
    Let $G_1$ over $\Sigma=\{a,b,c\}$ depicted in Fig.~\ref{fig:lbo-iso-ex} (left) be the instance of the LBO problem with the secret language $L_S=abb^*$ and the non-secret language $L_{NS}=acb^*$. Our transformation of LBO to ISO then results in the DES $G_1'$ depicted in Fig.~\ref{fig:lbo-iso-ex} (right) with a new observable event $@$, a single secret initial state~1, and a single non-secret initial state~4.
    We distinguish two cases depending on whether event $c$ is observable or not.

    In the first case, we assume that event $c$ is unobservable. In this case, $G_1$ is language-based opaque, because $P(L_S) \subseteq P(L_{NS})$, and the reader can see that $P(L(G_1',1)) = \overline{abb^*@} \subseteq \overline{ab^*@} = P(L(G_1',4))$. Therefore, $G_1'$ is initial-state opaque.

    In the second case, we assume that event $c$ is observable. In this case, $G_1$ is not language-based opaque, because $ab\in P(L_S)$ whereas $ab\not\in P(L_{NS})$, and we can see that $ab\in L(G_1',1)$ and $ab\not\in L(G_1',4)$. Therefore, $G_1'$ is not initial-state opaque.
    \begin{figure}
      \begin{minipage}{\textwidth}
        \centering
        \includegraphics[align=c,scale=.80]{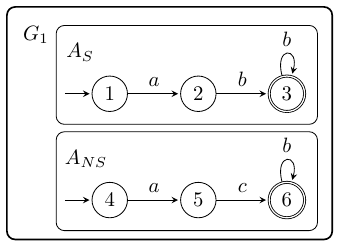}
        \quad $\Longrightarrow$ \quad
        \includegraphics[align=c,scale=.80]{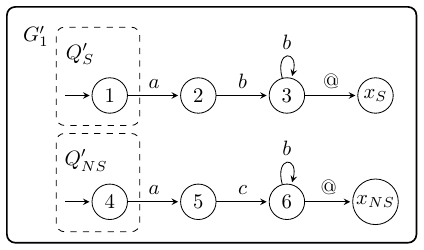}
      \end{minipage}
      \caption{An example of the transformation of the LBO problem (left) to the ISO problem (right).}
      \label{fig:lbo-iso-ex}
    \end{figure}
  \end{example}

\subsubsection{The case of a single observable event}
  The second step of our construction, Lemma~\ref{binEnc}, requires that $G_{ISO}$ has at least three observable events or, equivalently, that $G_{LBO}$ has at least two observable events. Consequently, our transformation does not preserve the number of observable events if $G_{LBO}$ has a single observable event. In fact, we show that there does not exist such a transformation unless \PTime $=$ \NP, which is a longstanding open problem of computer science. Deciding language-based opacity for systems with a single observable event is \coNP-complete~\cite{HolzerK11,StockmeyerM73}. We show that deciding initial-state opacity for systems with a single observable event is \NL-complete, and hence efficiently solvable on a parallel computer~\cite{AroraBarak2009}. In particular, the problem can be solved in polynomial time.

  \begin{theorem}\label{NLunaryISO}
    Deciding initial-state opacity for DESs with a single observable event is \NL-complete.
  \end{theorem}
  \begin{proof}
    Deciding initial-state opacity is equivalent to checking the inclusion of two prefix-closed languages. Namely, a DES $G$ with $\Sigma_o=\{a\}$ is initial-state opaque with respect to secret states $Q_S$ and non-secret states $Q_{NS}$ if and only if $K_S\subseteq K_{NS}$ for $K_S=P(L(G,Q_S))$ and $K_{NS}=P(L(G,Q_{NS}))$. Since the languages $K_S$ and $K_{NS}$ are prefix-closed, they are either finite, consisting of at most $|Q|$ strings, or equal to $\{a\}^*$.

    To show that the problem belongs to \NL, we show how to verify $K_S \not\subseteq K_{NS}$ in nondeterministic logarithmic space. Then, since \NL is closed under complement~\cite{Immerman88,Szelepcsenyi87}, $K_S \subseteq K_{NS}$ belongs to \NL. Thus, to check that $K_S \not\subseteq K_{NS}$ in nondeterministic logarithmic space, we guess $k \in \{0,\ldots, |Q|\}$ in binary, store it in logarithmic space, and verify that $a^k\in K_S$ and $a^k\notin K_{NS}$. To verify $a^k \in K_S$, we guess a path in $G$ step by step, storing only the current state, and counting the number of steps by decreasing $k$ by one in each step; logarithmic space is sufficient for this. Since $a^k \notin K_{NS}$ belongs to the complement of \NL, which coincides with \NL, we can check $a^k \notin K_{NS}$ in nondeterministic logarithmic space as well.

    To show that deciding initial-state opacity for DESs with a single observable event is \NL-hard, we reduce the DAG reachability problem~\cite{Jones75}: given a DAG $G=(V,E)$ and nodes $s,t\in V$, the problem asks whether $t$ is reachable from $s$.
    From $G$, we construct a DES $\A=(V\cup\{i\},\{a\},\delta,\{s,i\})$, where $i$ is a new initial state and $a$ is an observable event, as follows. With each node of $G$, we associate a state in $\A$. Whenever there is an edge from $j$ to $k$ in $G$, we add an $a$-transition from $j$ to $k$ to $\A$. We add a self-loop labeled by $a$ to state $t$ and to state $i$. The set of secret initial states is $Q_S=\{i\}$ and the set of non-secret initial states $Q_{NS}=\{s\}$. Then, $\A$ is initial-state opaque if and only if there is a path from $s$ to $t$ in $G$.
    Indeed, $L(\A,i)=\{a\}^*$ is included in $L(\A,s)$ if and only if $L(\A,s)=\{a\}^*$, which is if and only if $t$ is reachable from $s$.
  \qed\end{proof}

\subsubsection{Algorithmic complexity of deciding LBO}
  The algorithmic complexity of deciding whether a given DES is language-based opaque with respect to given secret and non-secret languages has been investigated in the literature. Lin~\cite{Lin2011} suggested an algorithm with the complexity $O(2^{2n})$, where $n$ is the order of the state spaces of the automata representing the secret and non-secret languages. The same complexity has been achieved by Wu and Lafortune~\cite{WuLafortune2013} using the transformation to current-state opacity. We improve this complexity.

  \begin{theorem}\label{lboAlgo}
    The time complexity of deciding whether a DES $G$ is language-based opaque with respect to a projection $P$, a secret language $L_S \subseteq L(G)$, and a non-secret language $L_{NS} \subseteq L(G)$ is $O(m\ell 2^{n_2} + n_12^{n_2})$, where $n_1$ is the number of states of the automaton recognizing $L_S$, $n_2$ is the number of states recognizing $L_{NS}$, $m\le \ell n_1^2$ is the number of transitions of an NFA recognizing $P(L_S)$, and $\ell$ is the number of observable events.
  \end{theorem}
  \begin{proof}
    Let $G_S$ and $G_{NS}$ be automata recognizing $L_S$ and $L_{NS}$ with $n_1$ and $n_2$ states, respectively. Then $P(L_S)\subseteq P(L_{NS})$ if and only if $P(L_S) \cap \text{co-}P(L_{NS}) = \emptyset$, where $\text{co-}P(L_{NS})$ stands for $\Sigma^* - P(L_{NS})$.
    We represent $P(L_S)$ by the projected automaton $P(G_S)$ with $m$ transitions and at most $n_1$ states, and $\text{co-}P(L_{NS})$ by the complement of the observer of $G_{NS}$, denoted by co-$G_{NS}^{obs}$, which has at most $2^{n_2}$ states and $\ell 2^{n_2}$ transitions.
    The problem is now equivalent to checking whether the language of $P(G_S) \cap \text{co-}G_{NS}^{obs}$ is empty, which means to search the structure for a reachable marked state. Since $P(G_S)$ has at most $n_1$ states and $m\le \ell n_1^2$ transitions, the structure has $O(m\ell 2^{n_2} + n_12^{n_2})$ transitions and states, which completes the proof.
  \qed\end{proof}

\subsection{Transformations between CSO and INSO}\label{redINSO-CSO}
  In this section, we provide the transformations between current-state opacity and infinite-step opacity. To the best of our knowledge, no transformations between current-state opacity and infinite-step opacity have been discussed in the literature so far.

\subsubsection{Transforming CSO to INSO}\label{redCSOtoINSO}
  We first focus on the transformation from current-state opacity to infinite-step opacity. The problem of deciding current-state opacity consists of a DES $G_{CSO}=(Q, \Sigma, \delta, I)$, a projection $P\colon\Sigma^*\to \Sigma_o^*$, a set of secret states $Q_{S}\subseteq Q$, and a set of non-secret states $Q_{NS} \subseteq Q$.
  From $G_{CSO}$, we construct a DES $G_{INSO}$ over the alphabet $\Sigma\cup\{u\}$, where $u$ is a new unobservable event. Specifically, we construct $G_{INSO}=(Q\cup\{q^\star\}, \Sigma\cup\{u\}, \delta', I)$ from $G_{CSO}$ by adding a new state $q^\star$ that is neither secret nor non-secret, and by defining $\delta'$ as follows, see Fig.~\ref{fig:cso-inso} for an illustration:
  \begin{enumerate}
    \item $\delta' = \delta$, that is, $\delta'$ is initialized as $\delta$ and further extended as follows;
    \item for each state $q\in Q_{NS}$, we add a transition $(q,u,q^\star)$ to $\delta'$;
    \item for each $a\in \Sigma$, we add a self-loop $(q^\star,a,q^\star)$ to $\delta'$.
  \end{enumerate}
  We extend the projection $P$ to the projection $P'\colon (\Sigma\cup\{u\})^* \to \Sigma_o^*$. The sets $Q_S$ and $Q_{NS}$ remain unchanged.

  \begin{figure}
    \begin{minipage}{\textwidth}
      \centering
      \includegraphics[align=c,scale=.89]{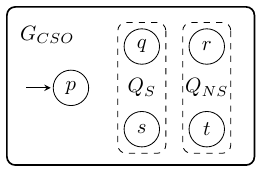}
      \quad $\Longrightarrow$ \quad
      \includegraphics[align=c,scale=.89]{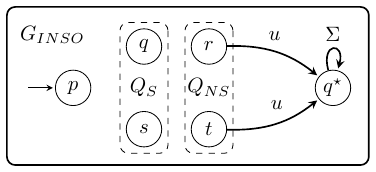}
    \end{minipage}
    \caption{Transforming CSO to INSO.}
    \label{fig:cso-inso}
  \end{figure}

  Notice that the transformation preserves the number of observable events and determinism, and that it requires only logarithmic space. It remains to show that $G_{CSO}$ is current-state opaque if and only if $G_{INSO}$ is infinite-step opaque.

  \begin{theorem}\label{thm:cso-inso}
    The DES $G_{CSO}$ is current-state opaque with respect to $Q_S$, $Q_{NS}$, and $P$ if and only if the DES $G_{INSO}$ is infinite-step opaque with respect to $Q_S$, $Q_{NS}$, and $P'$.
  \end{theorem}
  \begin{proof}
    Assume first that $G_{CSO}$ is not current-state opaque. Since the new state $q^\star$ is neither secret nor non-secret, we have that $G_{INSO}$ is not current-state opaque either. Consequently, $G_{INSO}$ is not infinite-step opaque.

    On the other hand, assume that $G_{CSO}$ is current-state opaque. Since the new state $q^\star$ is neither secret nor non-secret, we have that $G_{INSO}$ is current-state opaque as well. Let $st \in L(G_{INSO})$ be such that $\delta'(\delta'(I, s)\cap Q_S, t) \neq \emptyset$; in particular, $\delta'(I, s)\cap Q_S \neq \emptyset$. Then, since $G_{INSO}$ is current-state opaque, there exists $s'\in L(G_{INSO})$ such that $P'(s')=P'(s)$ and $\delta'(I,s')\cap Q_{NS} \neq\emptyset$. By construction, $s'$ can be extended by the string $ut$ using the transitions to state $q^\star$ followed by self-loops in state $q^\star$. Therefore, $\delta' (\delta'(I,s')\cap Q_{NS}, ut) \neq \emptyset$ and $P'(st) = P'(sut)$, which shows that $G_{INSO}$ is infinite-step opaque.
  \qed\end{proof}

  We now illustrate the construction in the following example.
  \begin{example}\label{ex:cso-inso}
    Let $G_2$ over $\Sigma=\{a,b,c\}$ depicted in Fig.~\ref{fig:cso-inso-ex} (left) be the instance of the CSO problem with the set of secret states $Q_S=\{2\}$ and the set of non-secret states $Q_{NS}=\{5\}$. Our transformation of CSO to INSO then results in the DES $G_2'$ depicted in Fig.~\ref{fig:cso-inso-ex} (right) with a new state $q^\star$ and a new unobservable event $u$.
    We distinguish two cases depending on whether event $c$ is observable or not.

    If event $c$ is unobservable, then $G_2$ is current-state opaque, because the only string leading to a secret state, state $2$, is the string $a$, for which the string $ac$ leading to the non-secret state, state $5$, satisfies that $P(a)=P(ac)$. Then, the reader can see that $G_2'$ is infinite-step opaque, because the only possible extensions of the string $a$ from the secret state $2$ are of the form $b^k$, for $k\in\mathbb{N}$, and for every such extension there is an extension $ub^k$ of the string $ac$ from the non-secret state $5$ such that $P(ab^k)=P(acub^k)$.

    If event $c$ is observable, then $G_2$ is not current-state opaque, because the only string leading to a non-secret state, string $ac$, has a different observation then the string $a$ leading to the secret state, that is, $P(ac)\neq P(a)$. Consequently, the reader can verify that $G_2'$ is not current-state opaque, and hence neither infinite-step opaque.
    \begin{figure}
      \begin{minipage}{\textwidth}
        \centering
        \includegraphics[align=c,scale=.89]{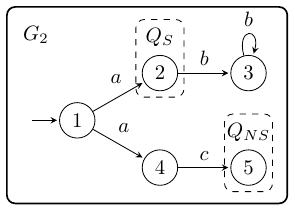}
        \quad $\Longrightarrow$ \quad
        \includegraphics[align=c,scale=.89]{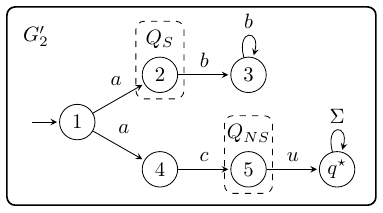}
      \end{minipage}
      \caption{
      An example of the transformation of the CSO problem (left) to the INSO problem (right).}
      \label{fig:cso-inso-ex}
    \end{figure}
  \end{example}

\subsubsection{Transforming INSO to CSO}\label{redINSOtoCSO}
  Transforming infinite-step opacity to current-state opacity is technically more involved. The problem of deciding infinite-step opacity consists of a DES $G_{INSO}=(Q, \Sigma, \delta, I)$, a projection $P\colon \Sigma^*\rightarrow \Sigma_o^*$, a set of secret states $Q_S\subseteq Q$, and a set of non-secret states $Q_{NS}\subseteq Q$.
  From $G_{INSO}$, we construct a DES $G_{CSO}$ in the following two steps:
  \begin{enumerate}
    \item We construct a DES $G_{CSO}$ such that $G_{CSO}$ is current-state opaque if and only if $G_{INSO}$ is infinite-step opaque. In this step of the construction, $G_{CSO}$ has one observable event more than $G_{INSO}$.

    \item To reduce the number of observable events by one, we apply Lemma~\ref{binEnc}. Consequently, the resulting DES has the same number of observable events as $G_{INSO}$, if $G_{INSO}$ has at least two observable events, is deterministic if and only if $G_{CSO}$ is, and is current-state opaque if and only if $G_{CSO}$ is.
  \end{enumerate}

  We now describe the construction of $G_{CSO}=(Q\cup Q^+\cup Q^-,\Sigma\cup\{@\},\delta',I)$, where $Q^+=\{q^+ \mid q \in Q\}$, $Q^-=\{q^- \mid q\in Q\}$, and $@$ is a new observable event. To this end, we first make two disjoint copies of $G_{INSO}$, denoted by $G_S$ and $G_{NS}$, where the set of states of $G_S$ is denoted by $Q_S'=Q^+$ and the set of states of $G_{NS}$ is denoted by $Q_{NS}'=Q^-$. The DES $G_{CSO}$ is taken as the disjoint union of the automata $G_{INSO}$, $G_S$, and $G_{NS}$, see Fig.~\ref{fig:inso-cso} for an illustration. Furthermore, for every state $q\in Q_S$, we add the transition $(q,@,q^+)$ and, for every state $q\in Q_{NS}$, we add the transition $(q,@,q^-)$. The set of secret states of $G_{CSO}$ is $Q_S'$ and the set of non-secret states of $G_{CSO}$ is $Q_{NS}'$. We extend projection $P$ to $P'\colon (\Sigma\cup\{@\})^* \to (\Sigma_o\cup\{@\})^*$.

  \begin{figure}
    \begin{minipage}{\textwidth}
      \centering
      \includegraphics[align=c,scale=.89]{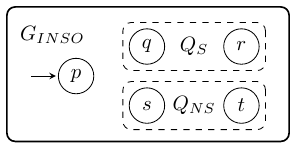}
      \quad $\Longrightarrow$ \quad
      \includegraphics[align=c,scale=.89]{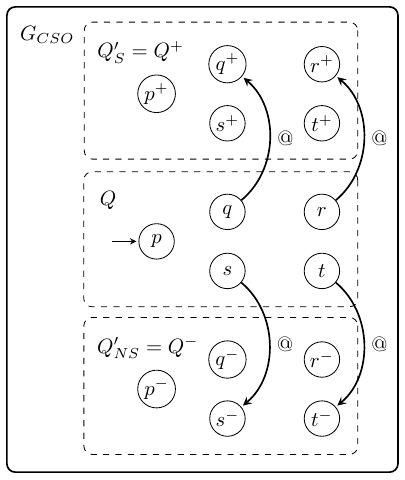}
    \end{minipage}
    \caption{Transforming INSO to CSO.}
    \label{fig:inso-cso}
  \end{figure}

  Notice that $G_{CSO}$ is deterministic if and only if $G_{INSO}$ is, and that logarithmic space is sufficient for the construction of $G_{CSO}$. As already pointed out, however, the construction does not preserve the number of observable events, which requires the second step of the construction using Lemma~\ref{binEnc} as described above.

  We now show that $G_{INSO}$ is infinite-step opaque if and only if $G_{CSO}$ is current-state opaque.

  \begin{theorem}\label{thm:inso-cso}
    The DES $G_{INSO}$ is infinite-step opaque with respect to $Q_S$, $Q_{NS}$, and $P$ if and only if the DES $G_{CSO}$ is current-state opaque with respect to $Q_S'$, $Q_{NS}'$, and $P'\colon (\Sigma\cup\{@\})^* \to (\Sigma_o\cup\{@\})^*$.
  \end{theorem}
  \begin{proof}
    Assume that $G_{INSO}$ is infinite-step opaque. We show that $G_{CSO}$ is current-state opaque. To this end, consider a string $w$ such that $\delta'(I,w)\cap Q_S' \neq \emptyset$. We want to show that there exists $w'$ such that $P'(w)=P'(w')$ and $\delta'(I,w')\cap Q_{NS}' \neq \emptyset$.
    However, since $Q_S'=Q^+$, $w$ is of the form $w_1 @ w_2$. Then, by the construction, $\delta(I,w_1)$ contains a secret state of $G_{INSO}$, say $q \in \delta(I,w_1)\cap Q_S$, such that state $q^+$ is a copy of state $q$ reached under $@$ from state $q$ in $G_{CSO}$, and $w_2$ is read from state $q^+$ in the copy $G_S$ of $G_{INSO}$. That is, $w_2$ can be read from state $q$ in $G_{INSO}$, and hence $\delta(I,w_1w_2)\neq\emptyset$. Altogether, $\delta(\delta(I,w_1)\cap Q_S, w_2) \neq \emptyset$ and the fact that $G_{INSO}$ is infinite-step opaque imply that there exists a string $w_1'w_2' \in L(G_{INSO})$ such that $P(w_1)=P(w_1')$, $P(w_2)=P(w_2')$, and $\delta (\delta(I,w_1')\cap Q_{NS}, w_2') \neq \emptyset$.
    Let $w'=w_1'@w_2'$. Then $P'(w)=P'(w')$ and, by the construction, $\emptyset \neq \delta'(\delta'(I,w_1'@)\cap Q_{NS}', w_2') \subseteq Q_{NS}'$, which completes the proof.

    On the other hand, assume that $G_{INSO}$ is not infinite-step opaque, that is, there exists a string $st \in L(G_{INSO})$ such that $\delta(\delta(I,s)\cap Q_S,t) \neq \emptyset$ and for every $s't' \in L(G_{INSO})$ with $P(s)=P(s')$ and $P(t)=P(t')$, $\delta (\delta(I,s')\cap Q_{NS}, t') = \emptyset$. But then for $s@t \in L(G_{CSO})$, we have that $\emptyset \neq \delta'(\delta'(I,s@)\cap Q_S',t) = \delta'(I,s@t) \subseteq Q_S'$ and, for every $s'@t' \in L(G_{CSO})$ such that $P'(s@t)=P'(s'@t')$, we have that $\delta'(I,s'@t') \cap Q_{NS}' = \delta' (\delta'(I,s'@)\cap Q_{NS}', t') = \emptyset$, which shows that $G_{CSO}$ is not current-state opaque.
  \qed\end{proof}

  We now illustrate the construction.
  \begin{example}\label{ex:inso-cso}
    Let $G_3$ over $\Sigma=\{a,b,c\}$ depicted in Fig.~\ref{fig:inso-cso-ex} (left) be the instance of the INSO problem with the set of secret states $Q_S=\{2\}$ and the set of non-secret states $Q_{NS}=\{4\}$. Our transformation of INSO to CSO then results in the DES $G_3'$ depicted in Fig.~\ref{fig:inso-cso-ex} (right) with a new observable event $@$, the set of secret states $Q_S'$, and the set of non-secret states $Q_{NS}'$.
    We again consider two cases based on the observability status of event $c$.

    If event $c$ is unobservable, then $G_3$ is infinite-step opaque. Indeed, the only string leading to the single secret state, state $2$, is the string $a$. The same string leads to the single non-secret state, state $4$. Then, any possible extension of the string $a$ from the unique secret state $2$ is the string $b^k$, for $k\in\mathbb{N}$, which reaches state $3$. However, for any such extension, there is the extension $cb^k$ from the non-secret state $4$ with $P(ab^k)=P(acb^k)$.
    The reader can further see that $G_3'$ is current-state opaque, because it can enter a secret state only after generating a string of the form $a@b^k$, $k\in\mathbb{N}$, in which case $\delta'(1,P^{-1}(a@))=\{2^+,4^-,5^-\}$ and $\delta'(1,P^{-1}(a@b^k))=\{3^+,5^-\}$ for $k\ge 1$.

    If event $c$ is observable, then $G_3$ is not infinite-step opaque, because after generating string $ab$, the intruder can deduce that the system was in the secret state $2$. Similarly, after generating string $a@b$, system $G_3'$ ends up in the only state $3^+$, which is a secret state, and hence $G_3'$ is not current-state opaque.
    \begin{figure}
      \begin{minipage}{\textwidth}
        \centering
        \includegraphics[align=c,scale=.89]{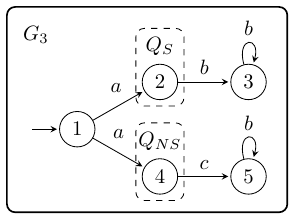}
        \quad $\Longrightarrow$ \quad
        \includegraphics[align=c,scale=.89]{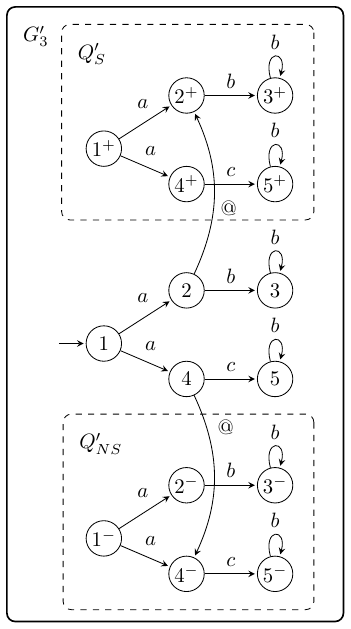}
      \end{minipage}
      \caption{An example of the transformation of the INSO problem (left) to the CSO problem (right).
      }
      \label{fig:inso-cso-ex}
    \end{figure}
  \end{example}

\subsubsection{The case of a single observable event}
  To preserve the number of observable events, our transformation of infinite-step opacity to current state opacity relies on Lemma~\ref{binEnc}. This lemma requires at least two observable events in $G_{INSO}$, and hence it is not applicable to systems with a single observable event. For these systems, we provide a different transformation that requires to add at most a quadratic number of new states.

  The problem of deciding infinite-step opacity for systems with a single observable event consists of a DES $G_{INSO}=(Q,\Sigma, \delta, I)$ with $\Sigma_o=\{a\}$, a set of secret states $Q_{S}\subseteq Q$, a set of non-secret states $Q_{NS}\subseteq Q$, and a projection $P\colon\Sigma^*\to \{a\}^*$.
  We denote the number of states of $G_{INSO}$ by $n$, and define a function $\varphi\colon  Q \rightarrow \{0,\ldots, n\}$ that assigns, to every state $q$, the maximal number $k\in \{0,\ldots, n\}$ of observable steps that are possible from state $q$; formally, $\varphi(q)=\max\left\{ k\in \{0,\ldots,n\} \mid \delta(q,P^{-1}(a^k))\neq \emptyset\right\}$.

  From $G_{INSO}$, we construct a DES $G_{CSO}=(Q', \Sigma, \delta', I)$ as illustrated in Fig.~\ref{fig:inso-cso-single}, where $\delta'$ is initialized as $\delta$ and modified as follows. For every state $p\in Q$ with $\varphi(p) > 0$, we add $n$ new states $p_1,\ldots,p_n$ to $Q'$ and $n$ new transitions $(p,a,p_1)$ and $(p_i,a,p_{i+1})$, for $i=1,\dots, n-1$, to $\delta'$. Finally, we replace every transition $(p,a,r)$ in $\delta'$ by the transition $(p_n,a,r)$. Notice that the transformation requires to add at most $n^2$ states, and hence it can be done in polynomial time.
  Let $Q'_S=Q_S$ and $Q'_{NS}=Q_{NS}$. For every state $p\in Q_S$ with $\varphi(p) = k > 0$, we add the corresponding states $p_1,\ldots,p_k$ to $Q'_S$. Analogously, for $p\in Q_{NS}$ with $\varphi(p)=k>0$, we add $p_1,\ldots,p_k$ to $Q'_{NS}$.

  Notice that the transformation can be done in polynomial time, preserves the number of observable events, and determinism. However, whether the transformation can be done in logarithmic space is open. Even if the DES had no unobservable event, to determine whether $\varphi(\cdot)=n$ is equivalent to the detection of a cycle. The detection of a cycle is \NL-hard: We can reduce the DAG reachability problem as follows. Given a DAG $G$ and two nodes $s$ and $t$, we construct a DES $\G$ by associating a state with every node of $G$ and an $a$-transition with every edge of $G$. Finally, we add an $a$-transition from $t$ to $s$. Then $t$ is reachable from $s$ in $G$ if and only if $\G$ contains a cycle. Since it is an open problem whether \Log $=$ \NL, it is an open problem whether $\varphi$ can be computed in deterministic logarithmic space.

  \begin{figure}
    \begin{minipage}{\textwidth}
      \centering
      \includegraphics[align=c,scale=.75]{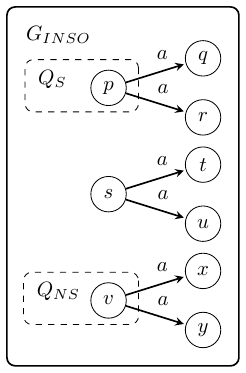}
      \quad $\Longrightarrow$ \quad
      \includegraphics[align=c,scale=.75]{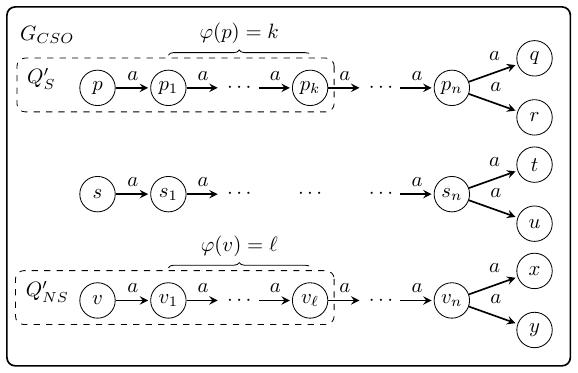}
    \end{minipage}
    \caption{Transforming INSO to CSO for systems with a single observable event.}
    \label{fig:inso-cso-single}
  \end{figure}

  We show that $G_{INSO}$ is infinite-step opaque if and only if $G_{CSO}$ is current-state opaque.

  \begin{theorem}
    The DES $G_{INSO}$ with a single observable event is infinite-step opaque with respect to $Q_S$, $Q_{NS}$, and $P$ if and only if the DES $G_{CSO}$ is current-state opaque with respect to $Q'_S$, $Q'_{NS}$, and $P$.
  \end{theorem}
  \begin{proof}
    Assume that $G_{INSO}$ is not infinite-step opaque. Then, there exists $st\in L(G_{INSO})$ with $\delta(\delta(I,s)\cap Q_S, t) \neq \emptyset$ such that $\delta(\delta(I,P^{-1}P(s))\cap Q_{NS},P^{-1}P(t)) = \emptyset$. Let $f\colon \Sigma^* \to \Sigma^*$ be a morphism such that $f(a)=a^{n+1}$ and $f(b)=b$, for $a\neq b \in \Sigma$. Then, by construction, $\delta(I,s)= \delta'(I,f(s))$, and hence $\delta'(I,f(s)) \cap Q_S' \neq \emptyset$.
    If $\delta(I,P^{-1}P(s))\cap Q_{NS} = \emptyset$, then $\delta(I,f(P^{-1}P(s)))\cap Q_{NS}' = \emptyset$ because $\delta(I,s') = \delta'(I,f(s'))$ for any $s'\in P^{-1}P(s)$, and $G_{CSO}$ is not current-state opaque.
    Otherwise, we denote by $q_s\in \delta(I,s)\cap Q_S$ and $q_{ns} \in \delta(I,P^{-1}P(s)) \cap Q_{NS}$ the states with maximal $\varphi(q_s)$ and $\varphi(q_{ns})$. Since $G_{INSO}$ is not infinite-step opaque, $\varphi(q_s) > \varphi(q_{ns})$.
    Then, in $G_{CSO}$, $q_s$ has exactly one outgoing observable transition and is followed by $\varphi(q_s)=k$ secret states, while $q_{ns}$ is followed by $\varphi(q_{ns})<k$ non-secret states.
    Therefore, $\delta'(I,f(s)a^k) \cap Q_S'\neq\emptyset$ and $\delta'(I,f(s')a^k) \cap Q_{NS}'=\emptyset$ for any $s'\in P^{-1}P(s)$, and hence $G_{CSO}$ is not current-state opaque.

    On the other hand, assume that $G_{INSO}$ is infinite-step opaque, and that $\delta'(I,w)\cap Q_S'\neq \emptyset$. We show that $\delta'(I,P^{-1}P(w))\cap Q_{NS} \neq \emptyset$. Consider a state $q_s \in \delta'(I,w) \cap Q_S'$ and a path $\pi$ in $G_{CSO}$ leading to $q_s$ under $w$. Denote by $p$ the last state of $\pi$ that corresponds to a state of $G_{INSO}$; that is, $p$ is not a new state added by the construction of $G_{CSO}$. Since $q_s\in Q_S'$, we have, by construction, that $p\in Q_S$. Then the choice of $p$ partitions $w=uv$, where $u$, read along the path $\pi$, leads to state $p$, and $v=a^\ell$ is a suffix of length $\ell \le n$. Let $u'$ be a string such that $f(u')=u$. Then $p\in \delta(I,u')\cap Q_S$. Since $\varphi(p)\ge \ell$, there exists $t$ such that $P(t)=a^\ell$ and $\delta(\delta(I,u')\cap Q_S,t) \neq \emptyset$ in $G_{INSO}$. Then infinite-step opacity of $G_{INSO}$ implies that there exists $u''$ and $t'$ such that $P(u')=P(u'')$, $P(t)=P(t')$, and $\delta(\delta(I,u'')\cap Q_{NS},t') \neq \emptyset$. In particular, there is a state $q_{ns}\in \delta(I,u'')\cap Q_{NS}$ with $\varphi(q_{ns}) \ge \ell$, and $\delta'(I,f(u''))\cap Q_{NS}' \neq \emptyset$. Therefore, $\delta'(I,f(u'')a^\ell)\cap Q_{NS}' \neq \emptyset$ and $P(f(u'')a^\ell) = P(uv) = P(w)$, which completes the proof.
  \qed\end{proof}

  We now illustrate the construction.
  \begin{example}\label{ex:inso-cso-single}
    Let $G_4$ over $\Sigma=\{a,u\}$ depicted in Fig.~\ref{fig:inso-cso-ex-single} (left) be the instance of the INSO problem with a single observable event $\Sigma_o = \{a\}$, the set of secret states $Q_S=\{1\}$, and the set of non-secret states $Q_{NS}=\{3\}$. Then, $\varphi(1)=\varphi(3)=3$, and our transformation of INSO to CSO results in the DES $G_4'$ depicted in Fig.~\ref{fig:inso-cso-ex-single} (right) with the set of secret states $Q_S'$ and the set of non-secret states $Q_{NS}'$.
    We consider two cases based on the presence of the unobservable transition $(1,u,3)$ in $G_4$.

    We first assume that the transition $(1,u,3)$ exists in $G_4$. Then, $G_4$ is infinite-step opaque, because any string $a^k$ leading from the secret state $1$ is indistinguishable from the string $ua^k$ that leads the system to the non-secret state $3$.
    The reader can see that $G_4'$ is current-state opaque, because a secret state is reachable only under a string of the form $a^k$, for $k\in\{0,1,2,3\}$, and for any such string there is an indistinguishable string $ua^k$ reaching a non-secret state.

    If the transition $(1,u,3)$ does not exist in $G_4$, then $G_4$ is not infinite-step opaque, because it is neither current-state opaque and, obviously, neither $G_4'$ is current-state opaque.

    \begin{figure}
      \begin{minipage}{\textwidth}
        \centering
        \includegraphics[align=c,scale=0.9]{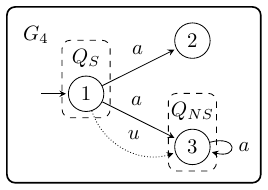}
        \quad $\Longrightarrow$ \quad
        \includegraphics[align=c,scale=0.9]{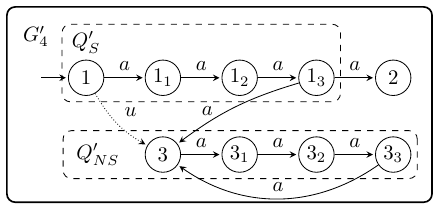}
      \end{minipage}
      \caption{An example of the transformation of the INSO problem with a single observable event (left) to the CSO problem (right).}
      \label{fig:inso-cso-ex-single}
    \end{figure}
  \end{example}

\subsubsection{Improving the algorithmic complexity of deciding infinite-step opacity}\label{insoAlgo}
  Let $G=(Q,\Sigma,\delta,I,F)$ be a DES. We design an algorithm deciding infinite-step opacity in time $O((n + m \ell) 2^n)$, where $\ell = |\Sigma_o|$ is the number of observable events, $n$ is the number of states of $G$, and $m$ is the number of transitions of $P(G)$, $m \le \ell n^2$.

  To decide whether $G$ is infinite-step opaque with respect to $Q_{S},Q_{NS}\subseteq Q$, and $P\colon \Sigma^* \to \Sigma_o^*$, we proceed as follows:
%

  \begin{enumerate}
   \item We compute the observer $\G^{obs}$ of $G$ in time $O(\ell 2^n)$~\cite{Lbook};

   \item We compute the projected automaton $P(G)$ of $G$ in time $O(m+n)$~\cite{HopcroftU79};

   \item We compute the product automaton $\C = P(G) \times \G^{obs}$ in time $O((m+n) \cdot \ell 2^n)$~\cite{Domaratzki2007};
    \begin{itemize}
      \item states of $\C$ are of the form $Q \times 2^Q$;
    \end{itemize}

    \item For every reachable state $X$ of $\G^{obs}$, we compute $X_S = X\cap Q_S$ and $X_{NS} = X \cap Q_{NS}$;
    \begin{enumerate}
      \item If $X_S\neq\emptyset$ and $X_{NS}=\emptyset$, then $G$ is not infinite-step opaque; this is, actually, the standard check whether $G$ is current-state opaque;

      \item Otherwise, for every state $x\in X_S$, we add a transition from $X$ under $@$ to state $(x,X_{NS})$ of $\C$, and we add the state $(x,X_{NS})$ to set $Y$;
    \end{enumerate}

    \item If $\C$ contains a state of the form $(q,\emptyset)$ reachable from $Y$, then $G$ is not infinite-step opaque; otherwise, $G$ is infinite-step opaque.
  \end{enumerate}

  Informally, we first make use of the standard check in the observer of $G$ whether $G$ is current-state opaque. If it is not, then it is neither infinite-step opaque. Otherwise, for every state $X$ of the observer of $G$ that contains both secret and non-secret states, we add a transition under the new event $@$ to a pair of a secret state $x\in X$ and the set of all non-secret states $X_{NS}$ of $X$. If a state of the form $(q,\emptyset)$ is reachable from $(x,X_{NS})$, then $G$ is not infinite-step opaque. Otherwise, $G$ is infinite-step opaque. We now formally prove correctness.

  \begin{lemma}
    The DES $G$ is infinite-step opaque if and only if $G$ is current-state opaque and no state of the form $(q,\emptyset)$ is reachable in $\C$ from the set $Y$.
  \end{lemma}
  \begin{proof}
    Assume that $G$ is not infinite-step opaque. Then, there exists $st\in L(G)$ such that $\delta(\delta(I,s)\cap Q_S, t) \neq \emptyset$ and $\delta(\delta(I,P^{-1}P(s))\cap Q_{NS},P^{-1}P(t)) = \emptyset$. There are two cases:
    (i) either $\delta(I,P^{-1}P(s))\cap Q_{NS} = \emptyset$, in which case $G$ is not current-state opaque, neither infinite-step opaque, and the algorithm detects this situation in the observer of $G$ on line 4(a),
    (ii) or $\delta(I,P^{-1}P(s))\cap Q_{NS} = Z \neq \emptyset$. In this case, $P(s)@$ leads from the observer of $G$ to the pairs $(\delta(I,P^{-1}P(s))\cap Q_S)\times \{Z\}$ of the NFA $\C$. Since $\delta(I,st)\neq\emptyset$, there exists $(z,Z)\in (\delta(I,P^{-1}P(s))\cap Q_S)\times \{Z\}$ such that $P(t)$ leads the projected automaton $P(G)$ from state $z$ to a state $q$. However, $\delta(Z,P^{-1}P(t))=\emptyset$ implies that $P(t)$ leads the observer of $G$ from state $Z$ to state $\emptyset$, and hence the pair $(q,\emptyset)$ is reachable in $\C$ from a state of $Y$.

    On the other hand, if $G$ is infinite-step opaque, then it is current-state opaque, and we show that no state of the form $(q,\emptyset)$ is reachable in $\C$ from a state of $Y$. For the sake of contradiction, assume that a state of the form $(q,\emptyset)$ is reachable in $\C$ from a state of $Y$. Then, there must be a string $s$ such that $P(s)$ reaches a state $X$ in the observer of $G$ such that $X_S = X \cap Q_S$ contains a state $z$, $X \cap Q_{NS} = Z \neq \emptyset$, there is a transition under $@$ from $X$ to the pair $(z,Z)$ of $\C$, and the NFA $\C$ reaches state $(q,\emptyset)$ from $(z,Z)$ under a string $w$. In particular, there must be a string $t\in P^{-1}(w)$ that moves $G$ from state $z$ to state $q$. But then $\delta(\delta(I,s)\cap Q_S,t)\neq\emptyset$, and $\delta(\delta(I,P^{-1}P(s))\cap Q_{NS},P^{-1}(w))=\emptyset$, which means that $G$ is not infinite-step opaque -- a contradiction.
  \qed\end{proof}

  Since our algorithm constructs and searches the NFA $\C$ that has $O(n2^n)$ states and $O(m \ell 2^n)$ transitions, the overall time complexity of our algorithm is $O((n + m \ell) 2^n)$.

  We now illustrate the procedure in the following example.
  \begin{example}\label{ex:alg1}
    We consider system $G_3$ of Example~\ref{ex:inso-cso} as depicted in Fig.~\ref{fig:inso-cso-ex} with all the events $a$, $b$, $c$ observable, the set of secret states $Q_S=\{2\}$, and the set of non-secret states $Q_{NS}=\{4\}$. Then $G_3$ is current-state opaque, but not infinite-step opaque. To show that $G_3$ is not infinite-step opaque, our algorithm works as follows. First, notice that $P(G_3)$ coincides with $G_3$, since there are no unobservable transitions in $G_3$. A relevant part of the observer of $G_3$ is depicted in Fig.~\ref{ex:alg1-obs} (left), and a relevant part of the automaton $C$, i.e., of the product of $P(G_3)$ with the observer of $G_3$, is depicted in Fig.~\ref{ex:alg1-obs} (right). The only reachable state of the observer that has a nonempty intersection with $Q_S=\{2\}$ is state $X=\{2,4\}$, resulting in $X_S=\{2\}$ and $X_{NS}=\{4\}$. The algorithm then creates an $@$-transition from state $X=\{2,4\}$ of the observer to state $(2,\{4\})$ of the product automaton $C$ (the dashed transition in Fig.~\ref{ex:alg1-obs}). Since state $(3,\emptyset)$ is reachable from state $(2,\{4\})$ in $C$, system $G_3$ is not infinite-step opaque; indeed, observing $ab$ in $G_3$, the intruder knows for sure that the system was in a secret state.
    \begin{figure}
      \begin{center}
        \includegraphics[scale=.85]{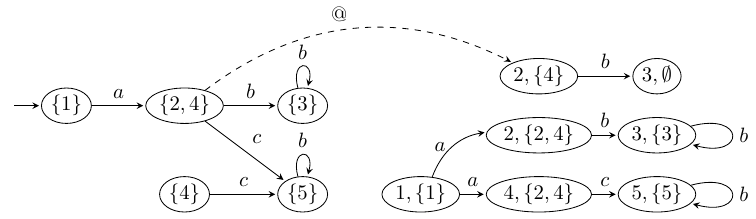}
        \caption{The relevant part of the observer of $G_3$ (left), the corresponding part of the automaton $C$ (right), and the $@$-transition (dashed) added by the algorithm.}
        \label{ex:alg1-obs}
      \end{center}
      \vspace{-15pt}
    \end{figure}
    
    {\makeatletter%
     \let\par\@@par\par%
     \parshape0\everypar{}%
    \begin{wrapfigure}{r}{0.38\textwidth}
      \centering
      \vspace{-10pt}
      \includegraphics[scale=.8]{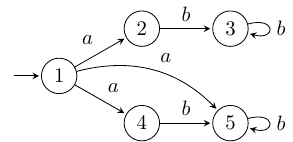}
      \caption{Projected automaton $P(\tilde{G_3})$.}
    \label{ex:alg1-obs2}
    \end{wrapfigure}
    On the other hand, we now assume that $c$ is unobservable. To avoid confusion, we denote $G_3$ with $a$ and $b$ observable, $c$ unobservable, the set of secret states $Q_S=\{2\}$, and the set of non-secret states $Q_{NS}=\{4\}$ as $\tilde{G_3}$. Then $\tilde{G_3}$ is infinite-step opaque, and our algorithm works as follows. First, we construct $P(\tilde{G_3})$ as shown in Fig.~\ref{ex:alg1-obs2}. Relevant parts of the observer of $\tilde{G_3}$ and of the product of $P(\tilde{G_3})$ with the observer, automaton $C$, is depicted in Fig.~\ref{ex:alg1-obs3}. 
    \begin{figure}[b]
      \centering
      \includegraphics[scale=.85]{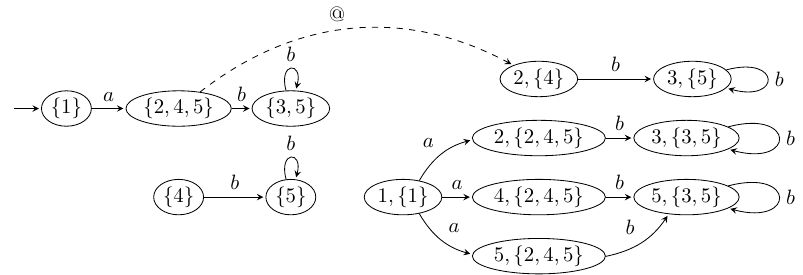}
      \caption{The relevant part of the observer of $\tilde{G_3}$ (left), the corresponding part of the automaton $C$ (right), and the $@$-transition (dashed) added by the algorithm.}
      \label{ex:alg1-obs3}
    \end{figure}
    The only reachable state of the observer of $\tilde{G_3}$ with a nonempty intersection with $Q_S=\{2\}$ is state $X=\{2,4,5\}$, resulting in $X_S=\{2\}$ and $X_{NS}=\{4\}$. The algorithm creates an $@$-transition from state $X=\{2,4,5\}$ of the observer to state $(2,\{4\})$ of the product automaton $C$ (the dashed transition in Fig.~\ref{ex:alg1-obs3}). Since no state of the form $(q,\emptyset)$ is reachable from state $(2,\{4\})$, $\tilde{G_3}$ is infinite-step opaque.\par}
  \end{example}

\subsection{Transformations between CSO and K-SO}\label{redKSO-CSO}
  In this section, we describe the transformations between current-state opacity and K-step opacity. To the best of our knowledge, no such transformations have been considered in the literature so far.

  \subsubsection{Transforming CSO to K-SO}\label{redCSOtoKSO}

  \begin{figure}
    \begin{minipage}{\textwidth}
      \centering
      \includegraphics[align=c,scale=.75]{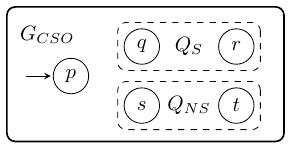}
      \quad $\Longrightarrow$ \quad
      \includegraphics[align=c,scale=.75]{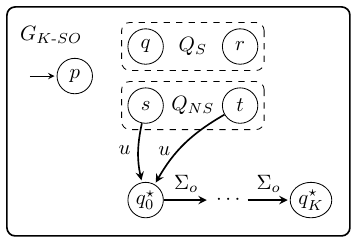}
    \end{minipage}
    \caption{Transforming CSO to K-SO.}
    \label{fig:cso-kso}
  \end{figure}

  The transformation from current state opacity to K-step opacity is analogous to the transformation from current state opacity to infinite-step opacity of Section~\ref{redCSOtoINSO}. Intuitively, the modification is that we need to make only K observable steps from any non-secret state instead of infinitely many such steps.

  The problem of deciding current-state opacity consists of a DES $G_{CSO}=(Q, \Sigma, \delta, I)$, a projection $P\colon\Sigma^*\to \Sigma_o^*$, a set of secret states $Q_{S}\subseteq Q$, and a set of non-secret states $Q_{NS} \subseteq Q$.
  For a given $K\in \mathbb{N}$, from $G_{CSO}$, we construct a DES $G$\textsubscript{$K$-$SO$}$=(Q\cup Q^\star, \Sigma\cup\{u\}, \delta', I)$, where $u$ is a new unobservable event, by adding $K+1$ new states $Q^\star=\{q_0^\star, \ldots, q_{K}^\star\}$ that are neither secret nor non-secret, and by defining $\delta'$ as follows, see Fig.~\ref{fig:cso-kso} for an illustration:
  \begin{enumerate}
  \item $\delta' = \delta$, that is, $\delta'$ is initialized as $\delta$ and further extended as follows;
  \item for every state $q\in Q_{NS}$, we add the transition $(q,u,q_0^\star)$ to $\delta'$;
  \item for $i=0,\ldots, K-1$ and every $a\in\Sigma_o$, we add the transition $(q_i^\star, a, q_{i+1}^\star)$ to $\delta'$.
  \end{enumerate}
  We extend the projection $P$ to the projection $P'\colon (\Sigma\cup\{u\})^* \to \Sigma_o^*$. The sets $Q_S$ and $Q_{NS}$ remain unchanged.

  \begin{theorem}\label{thm:cso-kso}
    The DES $G_{CSO}$ is current-state opaque with respect to $Q_S$, $Q_{NS}$, and $P$ if and only if the DES $G$\textsubscript{$K$-$SO$} is K-step opaque with respect to $Q_S$, $Q_{NS}$, $P'$, and $K$.
  \end{theorem}
  \begin{proof}
    Assume first that $G_{CSO}$ is not current-state opaque. Since the new states $q_0^\star,\ldots, q_{K}^\star$ are neither secret nor non-secret, $G$\textsubscript{$K$-$SO$} is not current-state opaque either, and hence $G$\textsubscript{$K$-$SO$} is not K-step opaque.

    On the other hand, assume that $G_{CSO}$ is current-state opaque. Since the new states $q_0^\star,\ldots, q_{K}^\star$ are neither secret nor non-secret, $G$\textsubscript{$K$-$SO$} is current-state opaque as well. Let $st \in L(G$\textsubscript{$K$-$SO$}$)$ be such that $|P(t)| \leq K$ and $\delta'(\delta'(I, s)\cap Q_S, t) \neq \emptyset$. Then, since $G$\textsubscript{$K$-$SO$} is current-state opaque, there is $s'\in P^{-1}P(s)$ such that $\delta'(I,s')\cap Q_{NS} \neq\emptyset$. By construction, we can extend $s'$ by the string $uP(t)$ using the transitions through the new states $q_0^\star,\ldots,q_{K}^\star$, that is, $\delta' (\delta'(I,s')\cap Q_{NS}, uP(t)) \neq \emptyset$, and hence $G$\textsubscript{$K$-$SO$} is K-step opaque. \qed
  \end{proof}

  We now illustrate the construction.
  \begin{example}\label{ex:cso-kso}
    Let $G_2$ over $\Sigma=\{a,b,c\}$ depicted in Fig.~\ref{fig:cso-kso-ex} (left) be the instance of the CSO problem from Example~\ref{ex:cso-inso} with the set of secret states $Q_S=\{2\}$ and the set of non-secret states $Q_{NS}=\{5\}$. Our transformation of CSO to K-SO then results in the DES $G_2''$ depicted in Fig.~\ref{fig:cso-kso-ex} (right) with $K=2$, a new unobservable event $u$, and three new states $q_0^\star$, $q_1^\star$, and $q_2^\star$.
    We again distinguish two cases depending on whether event $c$ is observable or not.

    If $c$ is unobservable, $G_2$ is current-state opaque as shown in Example~\ref{ex:cso-inso}. The reader can see that $G_2''$ is then 2-step opaque, because the only possible extensions of the string $a$ from the secret state $2$ are of the form $b^k$, for $k\in\mathbb{N}$, and for those extensions where $k\leq 2$, there is an extension $ub^k$ of the string $ac$ from the non-secret state $5$ such that $P(ab^k)=P(acub^k)$.

    If $c$ is observable, then $G_2$ is not current-state opaque as shown in Example~\ref{ex:cso-inso}. Consequently, the reader can verify that $G_2''$ is not current-state opaque, and hence neither 2-step opaque.
    \begin{figure}
      \begin{minipage}{\textwidth}
        \centering
        \includegraphics[align=c,scale=0.9]{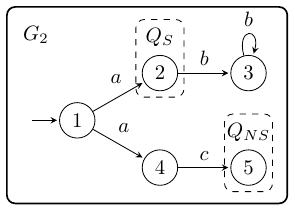}
        \quad $\Longrightarrow$ \quad
        \includegraphics[align=c,scale=0.9]{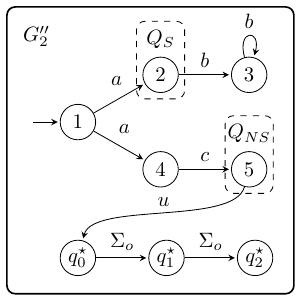}
      \end{minipage}
      \caption{An example of the transformation of the CSO problem (left) to the K-SO problem (right).}
      \label{fig:cso-kso-ex}
    \end{figure}
  \end{example}

\subsubsection{Transforming K-SO to CSO}\label{redKSOtoCSO}
  Transforming K-step opacity to current-state opacity is again similar to the transformation of infinite-step opacity to current-state opacity. Again, we only need to check K subsequent steps instead of all the subsequent steps.
  The problem of deciding K-step opacity consists of a DES $G_{\text{K-SO}}=(Q, \Sigma, \delta, I)$, a projection $P\colon \Sigma^*\rightarrow \Sigma_o^*$, a set of secret states $Q_S\subseteq Q$, and a set of non-secret states $Q_{NS}\subseteq Q$.
  From $G_{\text{K-SO}}$, we construct a DES $G_{CSO}$ in the following two steps:
  \begin{enumerate}
    \item We construct a DES $G_{CSO}$ such that $G_{CSO}$ is current-state opaque if and only if $G_{\text{K-SO}}$ is K-step opaque. In this step of the construction, $G_{CSO}$ has one observable event more than $G_{\text{K-SO}}$.

    \item To reduce the number of observable events by one, we apply Lemma~\ref{binEnc}. Consequently, the resulting DES has the same number of observable events as $G_{\text{K-SO}}$, if $G_{\text{K-SO}}$ has at least two observable events, is deterministic if and only if $G_{CSO}$ is, and is current-state opaque if and only if $G_{CSO}$ is.
  \end{enumerate}

  \begin{figure}
    \begin{minipage}{\textwidth}
      \centering
      \includegraphics[align=c,scale=.75]{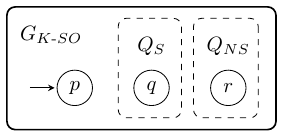}
      \quad $\Longrightarrow$ \quad
      \includegraphics[align=c,scale=.75]{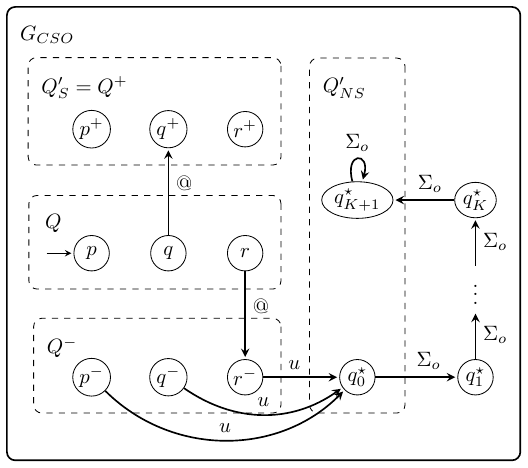}
    \end{minipage}
    \caption{Transforming K-SO to CSO.}
    \label{fig:k-so-cso}
  \end{figure}

  We now describe the construction of $G_{CSO}=(Q\cup Q^+\cup Q^-\cup Q^\star,\Sigma\cup\{u,@\},\delta',I)$, where
  $Q^+=\{q^+ \mid q \in Q\}$,
  $Q^-=\{q^- \mid q\in Q\}$,
  $Q^\star=\{q^\star_0,\ldots, q^\star_{K+1}\}$,
  $@$ is a new observable event, and $u$ is a new unobservable event.
  To this end, we first make two disjoint copies of $G_{\text{K-SO}}$, denoted by $G^+$ and $G^-$, where the set of states of $G^+$ is denoted by $Q^+$ and the set of states of $G^-$ is denoted by $Q^-$.
  The DES $G_{CSO}$ is now taken as the disjoint union of the automata $G_{\text{K-SO}}$, $G^+$, and $G^-$, see Fig.~\ref{fig:inso-cso} for an illustration.
  We now add K+2 new states $q_0^\star,\ldots, q_{K+1}^\star$ to $G_{CSO}$ and the following transitions.
  For every state $q\in Q_S$, we add the transition $(q,@,q^+)$, for every state $q\in Q_{NS}$, we add the transition $(q,@,q^-)$, for every $q^-\in Q^-$, we add the transition $(q^-, u,q_0^\star)$, for every $a\in\Sigma_o$ and $i=0,\dots, K$, we add the transition $(q_i^\star,a,q_{i+1}^\star)$, and, finally, we add the self-loop $(q_{K+1}^\star,a,q_{K+1}^\star)$ for every $a\in\Sigma_o$.
  The set of secret states of $G_{CSO}$ is the $Q_S'=Q^+$ and the set of non-secret states of $G_{CSO}$ is the set $Q_{NS}'=\{q_0^\star,q_{K+1}^\star\}$. We extend projection $P$ to $P'\colon (\Sigma\cup\{@, u\})^* \to (\Sigma_o\cup\{@\})^*$.

  Notice that $G_{CSO}$ is deterministic if and only if $G_{\text{K-SO}}$ is, and that logarithmic space is sufficient for the construction of $G_{CSO}$. However, as already pointed out, the construction does not preserve the number of observable events, which requires the second step of the construction using Lemma~\ref{binEnc}.

  We now show that $G_{\text{K-SO}}$ is K-step opaque if and only if $G_{CSO}$ is current-state opaque.

  \begin{theorem}\label{thm:kso-cso}
    The DES $G$\textsubscript{$K$-$SO$} is K-step opaque with respect to $Q_S$, $Q_{NS}$, and $P$ if and only if the DES $G_{CSO}$ is current-state opaque with respect to $Q_S'$, $Q_{NS}'$, and $P'\colon (\Sigma\cup\{@,u\})^* \to (\Sigma_o\cup\{@\})^*$.
  \end{theorem}
  \begin{proof}
    Assume that $G_{\text{K-SO}}$ is K-step opaque. We show that $G_{CSO}$ is current-state opaque. To this end, consider a string $w$ such that $\delta'(I,w)\cap Q_S' \neq \emptyset$. We want to show that there exists $w'\in P'^{-1}P'(w)$ such that $\delta'(I,w')\cap Q_{NS}' \neq \emptyset$.
    However, since $Q_S'=Q^+$, $w$ is of the form $w_1 @ w_2$ and, by the construction, $\delta(I,w_1)$ contains a secret state of $G_{\text{K-SO}}$.
    Since $G$ is K-step opaque, there exists a string $w_1' \in P^{-1}P(w_1)$ such that $\delta(I,w_1')\cap Q_{NS} \neq \emptyset$.
    Then, because $w_2$ can be read in the copy of $G_{\text{K-SO}}$ from a state $q^+$ for a state $q \in \delta(I,w_1)\cap Q_S$, we further have that $\delta(\delta(I,w_1)\cap Q_S, w_2) \neq \emptyset$.
    If $|P(w_2)|\le K$, then K-step opacity of $G_{\text{K-SO}}$ implies that there exists a string $w_1''w_2'' \in L(G_{\text{K-SO}})$ such that $P(w_1'')=P(w_1)$, $P(w_2'')=P(w_2)$, and $\delta (\delta(I,w_1'')\cap Q_{NS}, w_2'') \neq \emptyset$. By construction, $q_0^\star \in \delta' (\delta'(I,w_1''@)\cap Q_{NS}, w_2''u)$, and hence $G_{CSO}$ is current-state opaque.
    If $|P(w_2)|> K$, then $q_{K+1}^\star \in \delta' (\delta'(I,w_1''@)\cap Q_{NS}, uP(w_2''))$, and hence $G_{CSO}$ is current-state opaque.

    On the other hand, assume that $G_{\text{K-SO}}$ is not K-step opaque, that is, there exists a string $st \in L(G_{\text{K-SO}})$ such that $|P(t)|\le K$, $\delta(\delta(I,s)\cap Q_S,t) \neq \emptyset$ and, for every $s'\in P^{-1}P(s)$ and $t'\in P^{-1}P(t)$, $\delta (\delta(I,s')\cap Q_{NS}, t') = \emptyset$. But then, for $s@t \in L(G_{CSO})$, we have that $\delta'(\delta'(I,s@)\cap Q_S',t) \cap Q_S' \neq \emptyset$ and, for every $s'@t' \in L(G_{CSO})$ such that $P'(s@t)=P'(s'@t')$, we have two cases:
    (i) If $\delta(I,s') \cap Q_{NS} = \emptyset$, then $\delta'(I,s'@t') \cap Q_{NS}' = \delta' (\delta'(I,s'@)\cap Q^-, t') = \delta' (\emptyset, t') = \emptyset$, which shows that $G_{CSO}$ is not current-state opaque.
    (ii) If $\delta(I,s') \cap Q_{NS} \neq \emptyset$, then $\delta'(I,s'@t') \cap Q_{NS}' = \delta' (\delta'(I,s'@)\cap Q^-, t') = \emptyset$, because inserting $u$ to any strict prefix of $t'$ may reach $q_0^\star$ but has to leave it when the rest of $t'$ is read, and the rest (neither $P(t')$) is not long enough to reach state $q_{K+1}^\star$. Therefore, $G_{CSO}$ is not current-state opaque.
  \qed\end{proof}

  We now illustrate the construction.
  \begin{example}\label{ex:kso-cso}
    \begin{figure}
      \begin{minipage}{\textwidth}
        \centering
        \includegraphics[align=c,scale=0.9]{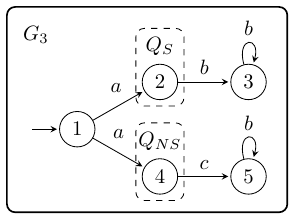}
        \quad $\Longrightarrow$ \quad
        \includegraphics[align=c,scale=0.9]{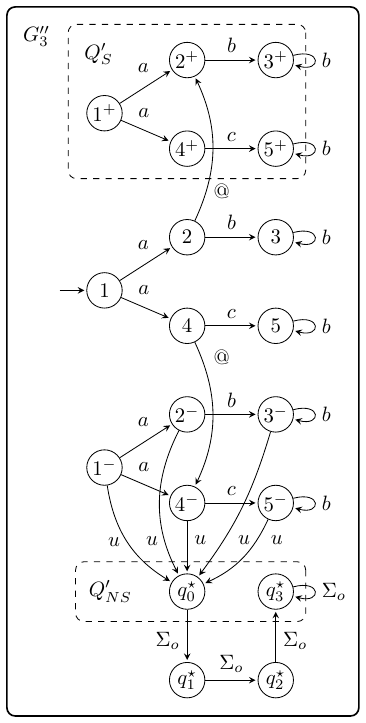}
      \end{minipage}
      \caption{An example of the transformation of the K-SO problem (left) to the CSO problem (right).}
      \label{fig:kso-cso-ex}
    \end{figure}
    Let $G_3$ over $\Sigma=\{a,b,c\}$ depicted in Fig.~\ref{fig:kso-cso-ex} (left) be the instance of the K-SO problem from Example~\ref{ex:inso-cso} with $K=2$, the set of secret states $Q_S=\{2\}$, and the set of non-secret states $Q_{NS}=\{4\}$. Our transformation of K-SO to CSO then results in the DES $G_3''$ depicted in Fig.~\ref{fig:kso-cso-ex} (right) with a new observable event $@$, a new unobservable event $u$, the set of secret states $Q_S'$, and the set of non-secret states $Q_{NS}'$.
     We consider two cases based on the observability status of event $c$.

    If $c$ is unobservable, then $G_3$ is 2-step opaque, because it is infinite-step opaque as shown in Example~\ref{ex:inso-cso}.
    The reader can further see that $G_3''$ is current-state opaque, because it can enter a secret state only after generating a string of the form $a@b^k$, for $k\in\mathbb{N}$, in which case we have that $\delta'(1,P^{-1}(a@))=\{2^+,4^-,5^-,q_0^\star\}$ and $\delta'(1,P^{-1}(a@b^k))=\{3^+,5^-,q_0^\star,\ldots,q_i^\star\}$ for $k\ge 1$, where $i=\min\{k, 3\}$.

    If $c$ is observable, then $G_3$ is not 2-step opaque, because after generating string $ab$, the intruder can deduce that the system was in the secret state $2$. Similarly, after generating string $a@b$, system $G_3''$ ends up in the only state $3^+$, which is a secret state, and hence $G_3''$ is not current-state opaque.
  \end{example}

\subsubsection{The case of a single observable event}
  To preserve the number of observable events, our transformation of K-step opacity to current state opacity relies on Lemma~\ref{binEnc}. This lemma requires at least two observable events in $G_{\text{K-SO}}$, and hence it is not applicable to systems with a single observable event. For these systems, we provide a different transformation that requires to add at most a quadratic number of new states.

  The problem of deciding K-step opacity for systems with a single observable event consists of a DES $G_{\text{K-SO}}=(Q,\Sigma, \delta, I)$ with $\Sigma_o=\{a\}$, a set of secret states $Q_{S}\subseteq Q$, a set of non-secret states $Q_{NS}\subseteq Q$, and a projection $P\colon\Sigma^*\to \{a\}^*$.
  We denote the number of states of $G_{\text{K-SO}}$ by $n$, and define a function $\varphi\colon  Q \rightarrow \{0,\ldots,K\}$ that assigns, to every state $q$, the maximal number $k\in \{0,\ldots,K\}$ of observable steps that are possible from state $q$; formally, $\varphi(q)=\max\left\{ k\in \{0,\ldots,K\} \mid \delta(q,P^{-1}(a^k))\neq \emptyset\right\}$. Notice that if $K > n-1$, then a system with a single observable event is K-step opaque if and only if it is infinite-step opaque. Therefore, we may consider only $K\le n-1$.

  From $G_{\text{K-SO}}$, we construct a DES $G_{CSO}=(Q', \Sigma, \delta', I)$ as illustrated in Fig.~\ref{fig:kso-cso-single}, where $\delta'$ is initialized as $\delta$ and modified as follows. For every state $p\in Q$ with $\varphi(p) > 0$, we add $K$ new states $p_1,\ldots,p_K$ to $Q'$ and $K$ new transitions $(p,a,p_1)$ and $(p_i,a,p_{i+1})$, for $i=1,\dots, K-1$, to $\delta'$. Finally, we replace every transition $(p,a,r)$ in $\delta'$ by the transition $(p_K,a,r)$. Notice that the transformation requires to add at most $n^2$ states, and hence it can be done in polynomial time.
  Let $Q'_S=Q_S$ and $Q'_{NS}=Q_{NS}$. For every state $p\in Q_S$ with $\varphi(p) = k > 0$, we add the corresponding states $p_1,\ldots,p_k$ to $Q'_S$ and, for every $p\in Q_{NS}$ with $\varphi(p)=k>0$, we add $p_1,\ldots,p_k$ to $Q'_{NS}$.

  Notice that the transformation can be done in polynomial time, preserves the number of observable events, and determinism. However, whether the transformation can be done in logarithmic space is open.

  \begin{figure}
    \begin{minipage}{\textwidth}
      \centering
      \includegraphics[align=c,scale=.75]{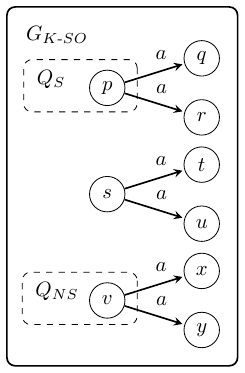}
      \quad $\Longrightarrow$ \quad
      \includegraphics[align=c,scale=.75]{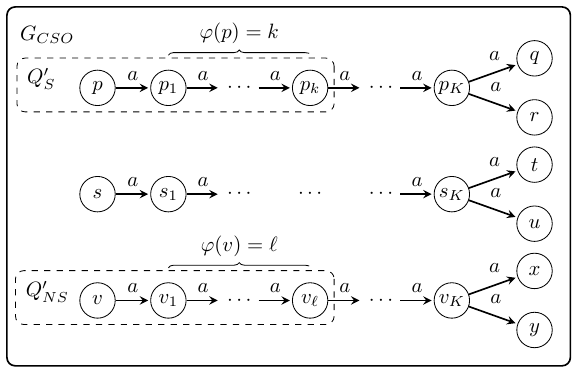}
    \end{minipage}
    \caption{Transforming K-SO to CSO for systems with a single observable event.}
    \label{fig:kso-cso-single}
  \end{figure}

  We show that $G_{\text{K-SO}}$ is K-step opaque if and only if $G_{CSO}$ is current-state opaque.

  \begin{theorem}
    The DES $G_{\text{K-SO}}$ with a single observable event is K-step opaque with respect to $Q_S$, $Q_{NS}$, and $P$ if and only if the DES $G_{CSO}$ is current-state opaque with respect to $Q'_S$, $Q'_{NS}$, and $P$.
  \end{theorem}
  \begin{proof}
    Assume that $G_{\text{K-SO}}$ is not K-step opaque, that is, there is $st\in L(G_{\text{K-SO}})$ with $|P(t)|\le K$ such that $\delta(\delta(I,s)\cap Q_S, t) \neq \emptyset$ and $\delta(\delta(I,P^{-1}P(s))\cap Q_{NS},P^{-1}P(t)) = \emptyset$. Let $f\colon \Sigma^* \to \Sigma^*$ be a morphism such that $f(a)=a^{K+1}$ and $f(b)=b$, for $a\neq b \in \Sigma$. Then, by construction, $\delta(I,s)= \delta'(I,f(s))$, and hence $\delta'(I,f(s)) \cap Q_S' \neq \emptyset$.
    If $\delta(I,P^{-1}P(s))\cap Q_{NS} = \emptyset$, then $\delta'(I,f(P^{-1}P(s)))\cap Q_{NS}' = \emptyset$ because $\delta(I,s') = \delta'(I,f(s'))$ for any $s'\in P^{-1}P(s)$, and $G_{CSO}$ is not current-state opaque.
    Otherwise, we denote by $q_s\in \delta(I,s)\cap Q_S$ and $q_{ns} \in \delta(I,P^{-1}P(s)) \cap Q_{NS}$ the states with maximal $\varphi(q_s)$ and $\varphi(q_{ns})$. Since $G_{\text{K-SO}}$ is not K-step opaque, $\varphi(q_s) > \varphi(q_{ns})$.
    Then, in $G_{CSO}$, $q_s$ has exactly one outgoing observable transition and is followed by $\varphi(q_s)=k$ secret states, while $q_{ns}$ is followed by $\varphi(q_{ns})<k$ non-secret states.
    Therefore, $\delta'(I,f(s)a^k) \cap Q_S'\neq\emptyset$ and $\delta'(I,f(s')a^k) \cap Q_{NS}'=\emptyset$ for any $s'\in P^{-1}P(s)$, and hence $G_{CSO}$ is not current-state opaque.

    On the other hand, assume that $G_{\text{K-SO}}$ is K-step opaque, and that $\delta'(I,w)\cap Q_S'\neq \emptyset$. We show that $\delta'(I,P^{-1}P(w))\cap Q_{NS} \neq \emptyset$. Consider a state $q_s \in \delta'(I,w) \cap Q_S'$ and a path $\pi$ in $G_{CSO}$ leading to $q_s$ under $w$. Denote by $p$ the last state of $\pi$ that corresponds to a state of $G_{\text{K-SO}}$; that is, $p$ is not a new state added by the construction of $G_{CSO}$. Since $q_s\in Q_S'$, we have, by construction, that $p\in Q_S$. Then the choice of $p$ partitions $w=uv$, where $u$, read along the path $\pi$, leads to state $p$, and $v=a^\ell$ is a suffix of length $\ell \le K$. Let $u'$ be a string such that $f(u')=u$. Then $p\in \delta(I,u')\cap Q_S$. Since $\varphi(p)\ge \ell$, there exists $t$ such that $P(t)=a^\ell$ and $\delta(\delta(I,u')\cap Q_S,t) \neq \emptyset$ in $G_{\text{K-SO}}$. Then K-step opacity of $G_{\text{K-SO}}$ implies that there exists $u''$ and $t'$ such that $P(u')=P(u'')$, $P(t)=P(t')$, and $\delta(\delta(I,u'')\cap Q_{NS},t') \neq \emptyset$. In particular, there is a state $q_{ns}\in \delta(I,u'')\cap Q_{NS}$ with $\varphi(q_{ns}) \ge \ell$, and $\delta'(I,f(u''))\cap Q_{NS}' \neq \emptyset$. Therefore, $\delta'(I,f(u'')a^\ell)\cap Q_{NS}' \neq \emptyset$ and $P(f(u'')a^\ell) = P(uv) = P(w)$, which completes the proof.
  \qed\end{proof}

  We now illustrate the construction.
  \begin{example}\label{ex:kso-cso-single}
    Let $G_4$ over $\Sigma=\{a,u\}$ depicted in Fig.~\ref{fig:kso-cso-ex-single} (left) be the instance of the K-SO problem from Example~\ref{ex:inso-cso-single} with $K=2$, a single observable event $\Sigma_o = \{a\}$, the set of secret states $Q_S=\{1\}$, and the set of non-secret states $Q_{NS}=\{3\}$. Then, $\varphi(1)=\varphi(3)=2$, and our transformation of K-SO to CSO results in the DES $G_4''$ depicted in Fig.~\ref{fig:kso-cso-ex-single} (right) with the set of secret states $Q_S'$ and the set of non-secret states $Q_{NS}'$.
    Analogously to Example~\ref{ex:inso-cso-single}, we consider two cases based on the presence of the unobservable transition $(1,u,3)$ in $G_4$.

    If the transition $(1,u,3)$ exists in $G_4$, then $G_4$ is 2-step opaque, since it is infinite-step opaque as shown in Example~\ref{ex:inso-cso-single}.
    The reader can see that $G_4''$ is current-state opaque, because a secret state is reachable only under a string of the form $a^k$ for $k\in\{0,1,2\}$, and for any such string there is an indistinguishable string $ua^k$ reaching a non-secret state.

    If the transition $(1,u,3)$ does not exist in $G_4$, then $G_4$ is not 2-step opaque, because it is neither current-state opaque and, obviously, neither $G_4''$ is current-state opaque.
    \begin{figure}
      \begin{minipage}{\textwidth}
        \centering
        \includegraphics[align=c,scale=0.9]{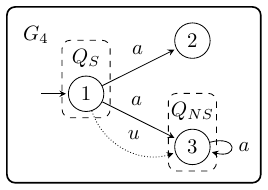}
        \quad $\Longrightarrow$ \quad
        \includegraphics[align=c,scale=0.9]{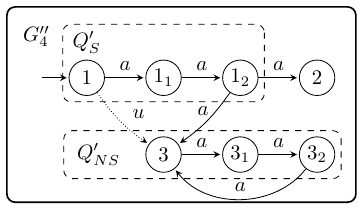}
      \end{minipage}
      \caption{An example of the transformation of the K-SO problem with a single observable event (left) to the CSO problem (right).}
      \label{fig:kso-cso-ex-single}
    \end{figure}
  \end{example}

\subsubsection{Improving the algorithmic complexity of deciding K-step opacity}\label{ksoAlgo}
  Let $G=(Q,\Sigma,\delta,I,F)$ be a DES. We design an algorithm deciding K-step opacity in time $O((K+1)2^n (n +\ell^2 m))$, where $\ell = |\Sigma_o|$ is the number of observable events, $n$ is the number of states of $G$, and $m$ is the number of transitions of $P(G)$, $m \le \ell n^2$.

  To decide whether $G$ is K-step opaque with respect to $Q_{S},Q_{NS}\subseteq Q$, and $P\colon \Sigma^* \to \Sigma_o^*$, we proceed as follows:

  \begin{enumerate}
   \item We compute the observer $\G^{obs}$ of $G$ in time $O(\ell 2^n)$;

   \item We compute the projected automaton $P(G)$ of $G$ in polynomial time $O(m+n)$;

   \item We compute a DFA $\D$ accepting the language $\Sigma_o^K$; then $\D$ has $K+1$ states and is constructed in time $O(\ell(K+1))$;

   \item We compute the product automaton $\C = P(G) \times \G^{obs}$ in time $O((m+n) \cdot \ell 2^n)$;
    \begin{itemize}
      \item states of $\C$ are of the form $Q\times 2^Q$;
    \end{itemize}

    \item For every reachable state $X$ of $\G^{obs}$, we compute $X_S = X\cap Q_S$ and $X_{NS} = X \cap Q_{NS}$;
    \begin{enumerate}
      \item If $X_S\neq\emptyset$ and $X_{NS}=\emptyset$, then $G$ is not K-step opaque;

      \item Otherwise, for every state $x\in X_S$, we add a transition from $X$ under $@$ to state $(x,X_{NS})$ of $\C$, and we add the state $(x,X_{NS})$ to set $Y$;
    \end{enumerate}

    \item We set $Y$ to be the set of initial states of $\C$, and compute $\G = \C \times \D$;

    \begin{enumerate}
      \item If $\G$ contains a reachable state of the form $(q,\emptyset,d)$, then $G$ is not K-step opaque; otherwise, $G$ is K-step opaque.
    \end{enumerate}
  \end{enumerate}

  Informally, we make use of the algorithm designed for deciding infinite-step opacity of Section~\ref{insoAlgo} with the modification that we take an intersection of $\C$ with the automaton recognizing $\Sigma_o^K$. This modification ensures that any computation of $\C$ ends after K steps, and hence we check at most K subsequent steps.

  \begin{lemma}
    The DES $G$ is K-step opaque if and only if $G$ is current-state opaque and no state of the form $(q,\emptyset,d)$ is reachable in $\G$.
  \end{lemma}
  \begin{proof}
    The algorithm works as that deciding infinite-step opacity. The only modification is that we intersect $\C$ with $\D$, recognizing $\Sigma_o^K$. This modification ensures that the algorithm checking infinite-step opacity is blocked after K subsequent steps, and hence it decides K-step opacity.
  \qed\end{proof}

  Since our algorithm constructs and searches the NFA $\G$ with $O((K+1)n2^n)$ states and $O((K+1)\ell m 2^n\ell)$ transitions, the time complexity of our algorithm is $O((K+1)2^n (n +\ell^2 m))$.

\section{Conclusions}
  We studied the transformations among the notions of language-based opacity, current-state opacity, initial-state opacity, initial-and-final-state opacity, K-step opacity, and infinite-step opacity. In particular, we provided a general transformation from language-based opacity to initial-state opacity, and constructed transformations between infinite-step opacity and current-state opacity, and between K-step opacity and current-state opacity. Together with the transformations of Wu and Lafortune~\cite{WuLafortune2013}, we have a complete list of transformations between the discussed notions of opacity. The transformations are computable in polynomial time, preserve the number of observable events, and determinism.
  We further applied the transformations to improve the algorithmic complexity of deciding language-based opacity, infinite-step opacity, and K-step opacity, and to obtain the precise computational complexity of deciding the discussed notions of opacity.

\begin{acknowledgements}
  We gratefully acknowledge suggestions and comments of the anonymous referees.
\end{acknowledgements}

\bibliographystyle{spmpsci}
\bibliography{mybib}

\end{document}